\pdfoutput=1
\documentclass[final,finalrunningheads,a4paper]{llncs}
\usepackage[utf8]{inputenc}

\DeclareMathSymbol{:}{\mathpunct}{operators}{"3A}

\usepackage{breakurl}
\usepackage{eucal}

\usepackage{centernot}
\usepackage{mdframed}

\usepackage{amsfonts} 
\usepackage{amssymb}

\usepackage{pifont}
\newcommand{\cmark}{\ding{51}}%
\newcommand{\xmark}{\ding{55}}%
\usepackage{multirow}

\usepackage{paralist}
\usepackage{mathtools}
\usepackage{url}

\usepackage{etoolbox}
\usepackage{xspace} 

\usepackage{stmaryrd} 

\usepackage[inline,shortlabels]{enumitem} 
\newlist{inlinelist}{enumerate*}{1}
\setlist*[inlinelist,1]{%
  label=(\roman*),
}

\usepackage{xifthen}  
\usepackage{ifthen}

\usepackage{nicefrac}
\usepackage{caption,subcaption}

\usepackage{color}

\usepackage{tikz} 
\usepackage{tikz-cd} 
\usepackage{float}

\usepackage{chngcntr}
\usepackage{apptools}
\AtAppendix{\counterwithin{definition}{section}}
\AtAppendix{\counterwithin{example}{section}}
\AtAppendix{\counterwithin{lemma}{section}}
\AtAppendix{\counterwithin{figure}{section}}

\usepackage{hyperref}
\usepackage{cleveref}

\usepackage[final,nomargin,inline,index]{fixme} 
\fxusetheme{color}

\definecolor{Black}{HTML}{000000}
\definecolor{Gray}{HTML}{808080}
\definecolor{Magenta}{HTML}{FF00FF}
\definecolor{RubineRed}{HTML}{ED017D}
\definecolor{ForestGreen}{HTML}{028A0F}
\definecolor{OliveGreen}{HTML}{808000}
\definecolor{MidnightBlue}{HTML}{006795}
\definecolor{Plum}{HTML}{92268F}

\FXRegisterAuthor{bart}{anbart}{\color{magenta} {\underline{bart}}}
\FXRegisterAuthor{zun}{anzun}{\color{OliveGreen} {\underline{zun}}}
\FXRegisterAuthor{ric}{anric}{\color{red} {\underline{ric}}}

\hypersetup{
  breaklinks   = true,
  colorlinks   = true, 
  urlcolor     = blue, 
  linkcolor    = blue, 
  citecolor    = red   
}
\hypersetup{final} 

\usepackage{listings}
\definecolor{listingBG}{HTML}{FFFFCB}%
\definecolor{listingFrame}{HTML}{BBBB98}%
\definecolor{listingLineno}{rgb}{0.5,0.5,1.0}%

\definecolor{LightGrey}{rgb}{0.975,0.975,0.975}

\lstdefinelanguage{txscript}{
	commentstyle=\color{Gray},
	morecomment=[l]{//},
	morecomment=[s]{/*}{*/},
	classoffset=0,
        escapeinside={(*}{*)},
	morekeywords={if,then,else,contract,skip,require,requireFinal,abort,return},
	keywordstyle=\color{Plum},
	classoffset=1,
	morekeywords={sender},
	keywordstyle=\addrColor,
	classoffset=2,        
	morekeywords={origin},
	keywordstyle=\pmvColor,
	classoffset=3,        
        morekeywords={constructor},
	keywordstyle=\color{MidnightBlue}\bfseries,
	basicstyle=\fontseries{m}\normalsize\ttfamily
	\lst@ifdisplaystyle\footnotesize\fi,
}

\newcommand{\change}[2]{{\marginpar{{\color{blue}\scriptsize\textbf{#1}}}}{{\color{blue}{#2}}}}%
\newcommand{\changeNoMargin}[1]{{\color{blue}{#1}}}%
\renewcommand{\change}[2]{#2}
\renewcommand{\changeNoMargin}[1]{#1}

\newcommand{\ifempty}[3]{%
  \ifthenelse{\isempty{#1}}{#2}{#3}%
}

\newcommand{\ifdots}[3]{%
  \ifthenelse{\equal{#1}{...}}{#2}{#3}%
}

\newtoggle{hidden}
\toggletrue{hidden}
\newcommand{\hidden}[1]{\iftoggle{hidden}{}{{\color{red}#1}}}

\newcommand{\keyterm}[1]{\textbf{\emph{#1}}}%

\makeatletter
\newcommand*{\itemequation}[3][]{%
  \item
  \begingroup
    \refstepcounter{equation}%
    \ifx\\#1\\%
    \else  
      \label{#1}%
    \fi
    \sbox0{#2}%
    \sbox2{$\displaystyle#3\m@th$}%
    \sbox4{\@eqnnum}%
    \dimen@=.5\dimexpr\linewidth-\wd2\relax
    \ifcase
        \ifdim\wd0>\dimen@
          \z@
        \else
          \ifdim\wd4>\dimen@
            \z@
          \else 
            \@ne
          \fi 
        \fi
      \@latex@warning{Equation is too large}%
    \fi
    \noindent   
    \rlap{\copy0}%
    \rlap{\hbox to \linewidth{\hfill\copy2\hfill}}%
    \hbox to \linewidth{\hfill\copy4}%
    \hspace{0pt}
  \endgroup
  \ignorespaces 
}
\makeatother

\newcommand{\Real}[1]{\mathrm{Real}}

\newcommand{\codefont}{\fontsize{9}{9}\selectfont}
\newcommand{\code}[1]{{\tt\codefont{#1}}}
\newcommand{\contract}[2][]{{\tt\codefont{\cmvColor{#2_{#1}}}}}
\newcommand{\txcode}[1]{{\tt\codefont{\txColor{#1}}}}

\newcommand{\sender}{{\addrColor{\code{sender}}}}
\newcommand{\origin}{{\pmvColor{\code{origin}}}}
\newcommand{\callee}[1]{{\cmvColor{\it callee}({#1})}}
\newcommand{\deps}[1]{{\cmvColor{\it deps}({#1})}}

\def\etc{etc.\@\xspace}
\newcommand{\Eg}{E.g.\@\xspace}
\newcommand{\eg}{e.g.\@\xspace}
\newcommand{\ie}{i.e.\@\xspace}
\newcommand{\wrt}{w.r.t.\@\xspace}

\renewcommand{\epsilon}{\varepsilon}
\newcommand{\emptyseq}{\epsilon}

\newcommand{\fst}{\mathit{fst}}

\def\negcaptionspace{\vspace{-10pt}}

\newenvironment{proofof}[2][]{%
  \ifempty{#1}
  {\subsection*{Proof of~\Cref{#2}}}
  {\subsection*{Proof of~\Cref{#2} ({#1})}}
  \label{#2-proof}
  }%
  {}

\def\addrColor{\color{magenta}}
\newcommand{\addrFmt}[1]{{\addrColor{\tt #1}}}
\newcommand{\addr}[2][]{\addrFmt{#2}_{\addrColor{#1}}\xspace}
\newcommand{\addrA}[1][]{\addr[{#1}]{A}} 
\newcommand{\AddrA}[1][]{{\addrColor{\mathcal{A}_{#1}}}} 
\newcommand{\AddrB}[1][]{{\addrColor{\mathcal{B}_{#1}}}} 

\newcommand{\AddrU}[1][]{\addrFmt{\mathbb{A}}_{\addrColor{#1}}} 

\newcommand{\cmvOfcst}[1]{\dag{#1}}
\newcommand{\strip}[2]{{#1}\!\upharpoonright_{#2}}

\def\pmvColor{\color{red}}
\newcommand{\pmvFmt}[1]{{\pmvColor{\tt #1}}}
\newcommand{\pmv}[2][]{\pmvFmt{#2}_{\pmvColor{#1}}\xspace}
\newcommand{\pmvA}[1][]{\pmv[{#1}]{A}} 
\newcommand{\pmvB}[1][]{\pmv[{#1}]{B}}
\newcommand{\pmvC}[1][]{\pmv[{#1}]{C}}
\newcommand{\pmvM}{\pmv{M}} 

\newcommand{\PmvM}[1][]{\pmvFmt{\mathcal{M}}_{\pmvColor{#1}}} 

\newcommand{\Adv}{\PmvM} 
\newcommand{\PmvU}{\pmvFmt{\mathbb{A}}_{\pmvColor{u}}} 

\def\cmvColor{\color{blue}}
\newcommand{\cmvFmt}[1]{{\cmvColor{\code{#1}}}}
\newcommand{\cmv}[2][]{\cmvFmt{#2}_{\cmvColor{#1}}}
\newcommand{\cmvC}[1][]{\cmv[#1]{C}} 
\newcommand{\cmvCi}[1][]{\cmvFmt{C'_{\cmvColor{{\rm #1}}}}}

\newcommand{\cmvD}[1][]{\cmv[#1]{D}} 
\newcommand{\cmvDi}[1][]{\cmvFmt{D'_{\cmvColor{{\rm #1}}}}}

\newcommand{\CmvC}[1][]{\cmv[#1]{\mathcal{C}}} 
\newcommand{\CmvCi}[1][]{\cmvFmt{\mathcal{C}'_{\cmvColor{{\rm #1}}}}}
\newcommand{\CmvD}[1][]{\cmv[#1]{\mathcal{D}}} 
\newcommand{\CmvDi}[1][]{\cmvFmt{\mathcal{D}'_{#1}}} 

\newcommand{\CmvU}[1][]{\cmvFmt{\mathbb{A}}_{\cmvColor{c}}} 

\def\cstColor{\color{MidnightBlue}}
\newcommand{\cstFmt}[1]{{\cstColor{#1}}}
\newcommand{\cst}[2][]{\cstFmt{#2}_{\cstColor{#1}}}
\newcommand{\cstC}[1][]{\cst[#1]{\Gamma}} 
\newcommand{\cstCi}[1][]{\cstFmt{\Gamma'_{\cstColor{{\rm #1}}}}}
\newcommand{\cstCii}[1][]{\cstFmt{\Gamma''_{\cstColor{{\rm #1}}}}}
\newcommand{\cstD}[1][]{\cst[#1]{\Delta}} 
\newcommand{\cstDi}[1][]{\cstFmt{\Delta'_{\cstColor{{\rm #1}}}}}

\def\txColor{\color{MidnightBlue}}

\newcommand{\txFmt}[1]{{\txColor{\sf #1}}}

\newcommand{\tx}[2][]{\txFmt{#2}_{\txColor{#1}}}
\newcommand{\txT}[1][]{\tx[#1]{X}} 
\newcommand{\txY}[1][]{\tx[#1]{Y}} 
\newcommand{\txTi}[1][]{\txFmt{X'_{\txColor{{\it #1}}}}}

\newcommand{\TxTS}[1][]{\vec{\tx[#1]{\mathcal{X}}}} 
\newcommand{\TxTiS}[1][]{\vec{\tx[#1]{\mathcal{X'}}}} 
\newcommand{\TxYiS}[1][]{\vec{\tx[#1]{\mathcal{Y'}}}} 
\newcommand{\TxYS}[1][]{\vec{\tx[#1]{\mathcal{Y}}}} 

\newcommand{\txin}[3]{{#1}?\,{#2:#3}}
\newcommand{\txout}[3]{{#1}!\,{#2:#3}}

\DeclareMathAlphabet{\mathbfsf}{\encodingdefault}{\sfdefault}{bx}{n}

\newcommand{\dom}[1]{\operatorname{dom} {#1}}

\newcommand{\Nat}{\mathbb{N}}

\newcommand{\setenum}[1]{\{#1\}}
\newcommand{\setcomp}[2]{\left\{{#1} \,\middle|\, {#2}\right\}}

\newcommand{\wmvA}[1][]{w_{#1}}
\newcommand{\wmvAi}[1][]{w'_{#1}}
\newcommand{\wmvAii}[1][]{w''_{#1}}
\newcommand{\WmvA}[1][]{W_{#1}}
\newcommand{\WmvAi}[1][]{W'_{#1}}
\newcommand{\WmvAii}[1][]{W''_{#1}}

\newcommand{\walE}[1]{\code{\#}{#1}}

\newcommand{\waltok}[2]{#1:#2}
\newcommand{\walenum}[1]{[#1]}
\newcommand{\walpmv}[2]{{#1}\walenum{#2}}
\newcommand{\walu}[3]{\walpmv{#1}{\waltok{#2}{#3}}}

\newcommand{\waldistrarrow}[1]{\approx_{\$}}

\newcommand{\wealth}[2]{\$_{#1}({#2})}
\newcommand{\idxfun}[1]{\mathbf{1}_{#1}}
\newcommand{\price}[1]{\$\idxfun{#1}}

\newcommand{\gain}[3]{\mathit{\gamma}_{#1}\ifempty{#2}{}{({#2},{#3})}}

\newcommand{\mall}[2]{{\kappa_{#1}\ifempty{#2}{}{({#2})}}}

\newcommand{\mev}[3]{{\mathrm{MEV\!}_{#1}\ifempty{#1#2#3}{}{({#2\ifempty{#3}{}{,#3}})}}}
\newcommand{\lmev}[3]{{\mathrm{MEV\!}_{#1}\ifempty{#1#2#3}{}{({#2\ifempty{#3}{}{,#3}})}}}
\newcommand{\rlmev}[3]{{\mathrm{MEV\!}^{\,\infty}_{#1}\ifempty{#1#2#3}{}{({#2\ifempty{#3}{}{,#3}})}}}

\newcommand{\nonintrel}{\mathrel{\not\rightsquigarrow}}
\newcommand{\nonint}[2]{\ifempty{#1}{\nonintrel}{{#1} \nonintrel {#2}}}
\newcommand{\negnonint}[2]{\ifempty{#1}{\rightsquigarrow}{{#1} \rightsquigarrow {#2}}}
\newcommand{\richnonintrel}{\mathrel{\not\rightsquigarrow^{\infty}}}
\newcommand{\richnonint}[2]{\ifempty{#1}{\richnonintrel}{{#1} \richnonintrel {#2}}}
\newcommand{\negrichnonint}[2]{\ifempty{#1}{\rightsquigarrow^{\infty}}{{#1} \rightsquigarrow^{\infty} {#2}}}

\newcommand{\intok}[2]{{\tokColor{\mathit{in}}}_{{#1}}({#2})}
\newcommand{\outtok}[2]{{\tokColor{\mathit{out}}}_{{#1}}({#2})}

\newcommand{\qedex}{\ensuremath{\diamond}}

\definecolor{LightGrey}{rgb}{0.95,0.95,0.95}
\definecolor{keyword}{HTML}{7F0055}

\def\tokColor{\color{ForestGreen}}
\newcommand{\tokFmt}[1]{{\tokColor{\tt #1}}}
\newcommand{\ETH}{\tokFmt{ETH}}
\newcommand{\tok}[2][]{\tokFmt{#2}_{\tokColor{#1}}\xspace}

\newcommand{\tokT}[1][]{\tok[{#1}]{T}}    
\newcommand{\tokTi}[1][]{\tok[{#1}]{T'}}
\newcommand{\TokU}{\tokFmt{\mathbb{T}}} 

\newlength\replength
\newcommand\repfrac{.1}

\newcommand\rulewidth{.6pt}
\newcommand\tdashfill[1][\repfrac]{\cleaders\hbox to \replength{%
  \smash{\rule[\arraystretch\ht\strutbox]{\repfrac\replength}{\rulewidth}}}\hfill}

\newcommand\tdotfill[1][\repfrac]{\cleaders\hbox to \replength{%
  \smash{\raisebox{\arraystretch\dimexpr\ht\strutbox-.1ex\relax}{.}}}\hfill}

\def\sysColor{\color{Black}}

\newcommand{\sysFmt}[1]{{\sysColor{#1}}}
\newcommand{\sysS}[1][]{\mathord{\sysFmt{S}_{\sysColor{#1}}}}

\newcommand{\sysSi}[1][]{\mathord{\sysColor{\sysS'_{#1}}}}

\newcommand{\sysSii}[1][]{\mathord{\sysColor{\sysS''_{#1}}}}

\newtoggle{anonymous}
\togglefalse{anonymous}

\newtoggle{arxiv}
\toggletrue{arxiv}

\makeatletter
\renewcommand\paragraph{\@startsection{paragraph}{4}{\z@}%
  {2.25ex \@plus 1ex \@minus .2ex}%
  {-0.75em}%
  {\normalfont\normalsize\bfseries}}
\makeatother

\begin{document}

\title{DeFi composability as MEV non-interference}

\iftoggle{anonymous}{}{
\author{Massimo Bartoletti\inst{1},
Riccardo Marchesin\inst{2},
Roberto Zunino\inst{2}}
\institute{
Universit\`a degli Studi di Cagliari, Cagliari, Italy
\and
Università di Trento, Trento, Italy
}
}


\maketitle

\begin{abstract}
  Complex DeFi services are usually constructed by
  composing a variety of simpler smart contracts.
  The permissionless nature of the blockchains where these
  smart contracts are executed
  makes DeFi services exposed to security risks,
  since adversaries can target any of the underlying contracts 
  to economically damage the compound service.
  We introduce a new notion of secure composability of smart contracts,
  which ensures that adversaries cannot economically harm the compound contract
  by interfering with its dependencies.
\end{abstract}

\section{Introduction}
\label{sec:intro}

Decentralized Finance (DeFi) is often touted
as the ``money Lego'' for its ability to build complex financial services
from a variety of simpler components~\cite{defipulse,Werner22aft}.
Recent empirical analyses of the DeFi ecosystem show that
DeFi protocols are heavily intertwined in practice,
giving rise to complex interactions~\cite{Kitzler22fc,Kitzler23tweb}.
This complexity has a drawback, in that adversaries can exploit
unintended forms of interaction among protocols
to obtain an economic profit to the detriment of the
users~\cite{Daian20flash,Gudgeon2020cvcbt}.
These security risks are further exacerbated by novel financial services
allowing users to easily create arbitrary compositions of
DeFi protocols~\cite{furucombo}.
%
%
The key unanswered question is:
\emph{when is a DeFi composition secure?}

Despite the clear practical relevance of this question,
the research on DeFi composability is surprisingly limited.
To the best of our knowledge, the only notion of secure DeFi composition
in the scientific literature is the one introduced by 
Babel, Daian, Kelkar and Juels in their recent
``Clockwork finance'' paper~\cite{Babel23clockwork}.
There, the focus is on attacks where adversaries can exploit
newly deployed contracts to increase their profit opportunities.
Accordingly, their criterion is that
contracts $\cstD$ are composable in a blockchain state $\sysS$
if adding $\cstD$ to $\sysS$
does not give the adversary a ``significantly higher''
Maximal Extractable Value (MEV).
In formulae, denoting by $\mev{}{\sysS}{}$ the 
maximal value that adversaries can extract from $\sysS$,
and with $\sysS \mid \cstD$ the blockchain state $\sysS$
extended with the contracts $\cstD$, 
the composability criterion of Babel \emph{et al.} is expressed as:
\begin{equation}
  \label{eq:babel-composability}
  \text{$\cstD$ is $\epsilon$-composable in $\sysS$}
  \quad \text{iff} \quad
  \mev{}{\sysS \mid \cstD}{} \leq (1 + \epsilon) \ \mev{}{\sysS}{}
\end{equation}
where $\epsilon$ parameterises the ``not significantly higher'' condition above.
For example, let $\cstD$ be a contract allowing users
to bet on the price of a token,
relying on an Automated Market Maker (AMM) in $\sysS$ as a price oracle.
If the adversary has sufficient capital in $\sysS$,
\eqref{eq:babel-composability} would correctly classify
$\cstD$ as \emph{not} $0$-composable in $\sysS$, because
the adversary can produce enough price volatility in the AMM
to always win the bet,
and so to extract more MEV in $\sysS \mid \cstD$ than in $\sysS$.

\paragraph{Limitations of $\epsilon$-composability}

We argue that the $\epsilon$-composability of Babel \emph{et al.} has some drawbacks.
First, using the $\mev{}{}{}$ of the \emph{whole} state $\sysS$
as a baseline for the comparison
makes it difficult to interpret the concrete security guarantee
of $\epsilon$-composability.
For instance, one may be induced to believe that
deploying a contract $\cstD$ in $\sysS$ is secure
after finding that \mbox{$0$-composability} holds.
However, \eqref{eq:babel-composability} only ensures that
the $\mev{}{}{}$ of \mbox{$\sysS \mid \cstD$} is bounded by that of $\sysS$,
while it does not say \emph{from which contracts} this MEV is extracted.
For instance,
assume that the Total Locked Value (TLV) of $\cstD$ and the MEV of $\sysS$ are equal,
and that extracting MEV from $\cstD$ blocks the MEV opportunities in $\sysS$
\change{B1}{(as in~\Cref{ex:babel-composability-not-implies-mev-nonintererence})}.
Even though $\cstD$ is $0$-composable with $\sysS$,
attacking \mbox{$\sysS \mid \cstD$} can make $\cstD$ lose its \emph{whole} TLV ---
not what the designer of $\cstD$ would reasonably consider secure.

Another drawback of using the \emph{global} MEV as a baseline
is that~\eqref{eq:babel-composability} classifies
as \emph{not} composable
contracts that have intended MEV, even when 
they have no interference at all with the rest of the system.
For instance, let $\cstD$ be an airdrop contract
allowing anyone to withdraw its whole balance.
Although $\cstD$ has no dependencies of any kind with the rest of the system,
it is not $0$-composable with any $\sysS$,
since $\mev{}{\sysS \mid \cstD}{}$ is greater than $\mev{}{\sysS}{}$.
Trying to find values of $\epsilon$ for which $\epsilon$-composability holds
is impractical, and at the extreme, if $\mev{}{\sysS}{}$ is zero,
then $\cstD$ is not $\epsilon$-composable with $\sysS$ for any $\epsilon \geq 0$.

Relying on global MEV also poses usability and algorithmic issues.
First, a definition of composability that requires to compute
a function of the \emph{whole} blockchain state $\sysS$
(which is usually quite large) is hardly efficiently computable.
Second, when $\epsilon$-composability is violated because of
additional MEV extracted from $\sysS$, it is unclear which countermeasures
could be taken against the attack, since the attacked contracts in $\sysS$
cannot be amended or removed.

\paragraph{A new security notion: MEV non-interference}

These arguments suggest to explore composability notions  
not relying on the global MEV.
Our insight comes from \emph{non-interference},
a notion studied by the information security community since the
1980s~\cite{GoguenMeseguer82sp,Ryan99csfw,Ryan01csfw,Backes02esorics,Giacobazzi04popl}.
In the classic setting, non-interference requires that
adversaries interacting with a software system cannot observe private data.
More precisely, the property holds when
public outputs (that can be observed by adversaries)
are not affected by private (confidential) inputs.
We adapt this notion to the DeFi setting,
by requiring that the MEV extractable from $\cstD$ (the public output)
is not affected by the dependencies of $\cstD$
(the private inputs).
This means that adversaries cannot extract more value from $\cstD$
using \emph{any} contract in $\sysS \mid \cstD$
than they could extract by using $\cstD$ only:
therefore, interacting with other contracts in the blockchain
(including the dependencies of $\cstD$)
gives no advantage to the adversary.
To the best of our knowledge, this is the first time
non-interference principles are applied to the DeFi setting.

Our new security notion, that we dub \emph{MEV non-interference},
relies on novel contributions to the theory of MEV.
In particular, we introduce \emph{local MEV},
a new metric of economic attacks to smart contracts
that measures the maximal economic loss that adversaries can cause
to a \emph{given} set of contracts
(whereas global MEV only applies to the \emph{whole} blockchain state).
We study two versions of local MEV, which assume different adversary models:
$\lmev{}{\sysS}{\CmvC}$ is the maximal loss of contracts $\CmvC$ in $\sysS$
under adversaries whose wealth is \emph{fixed} by $\sysS$;
instead, $\rlmev{}{\sysS}{\CmvC}$ assumes adversaries with an \emph{unbounded}
capital to carry the attack
(in practice, flash loans make this kind of adversary quite realistic).
One of our main theoretical contributions is that computing
$\rlmev{}{\sysS}{\CmvC}$ just requires to know
$\CmvC$ and their dependencies (\Cref{th:rich-lmev:stripping}).
This gives $\rlmev{}{}{}$ an algorithmic advantage over global MEV,
which depends on the whole blockchain state.

We define MEV non-interference in two versions:
\mbox{$\nonint{\sysS}{\cstD}$} means that 
interacting with $\sysS$ does not give \emph{bounded-wealth} adversaries
more opportunities to cause a loss to the contracts $\cstD$;
furthermore, \mbox{$\richnonint{\cstC}{\cstD}$}
does the same for \emph{unbounded-wealth} adversaries,
\change{B6}{where $\cstC$ is the state $\sysS$ after the wallets have been removed.}

We argue that MEV non-interference naturally captures the security property
a developer would like to verify before deploying new contracts.
By relying on the \emph{local} MEV that can be extracted from $\cstD$
(instead of the global MEV),
our notion overcomes the drawbacks of $\epsilon$-composability.
First, while contracts with intended MEV always break $\epsilon$-composability,
they possibly enjoy MEV non-interference:
\eg, this is the case of the above-mentioned airdrop contract,
which is MEV non-interfering \wrt any $\sysS$.
Second, for the unbounded-wealth version we prove that
deciding \mbox{$\richnonint{\cstC}{\cstD}$} only requires to consider
$\cstD$ and their dependencies in $\cstC$
(\Cref{th:richnonint:stripping}).
We exploit this result to show that $\richnonint{\cstC}{\cstD}$
is resistant to adversarial contracts: 
\ie, MEV non-interference still holds when
an adversary deploys some contracts $\cstDi$ before $\cstD$
(by contrast, both $\epsilon$-composability and $\nonint{}{}$
are \emph{not} resistant).
We give sufficient conditions for both versions of MEV non-interference
(\Cref{th:nonint:sufficient-conditions,th:richnonint:sufficient-conditions}),
and we apply them to study typical compositions of DeFi protocols
(\Cref{tab:defi-compositions}).


\iftoggle{arxiv}{}
{Because of space constraints, we provide the proofs of all our results
and the pseudo-code for our DeFi examples in a separated Appendix.}

\section{Blockchain model}
\label{sec:blockchain}

\begin{table}[t]
\caption{Summary of notation.}
\label{tab:notation}
\centering
\begin{tabular}{p{40pt}p{120pt}p{45pt}p{130pt}}
\hline
$\pmvA,\pmvB$ & User accounts 
& $\AddrA,\AddrB$ & Sets of [user$\mid$contract] accounts
\\
$\cmvC,\cmvD$ & Contract accounts 
& $\cmvOfcst{\cstC}$ & Contract accounts in $\cstC$
\\
$\CmvC,\CmvD$ & Sets of contract accounts
& $\deps{\CmvC}$ & Dependencies of $\CmvC$
\\
$\tokT,\tokTi$ & Token types 
& $\price{\tokT}$ & Price of $\tokT$
\\
$\txT,\txTi$ & Transactions 
& $\TxTS$ & Sequence of transactions
\\
$\sysS,\sysSi$ & Blockchain states 
& $\wealth{\AddrA}{\sysS}$ & Wealth of $\AddrA$ in $\sysS$
\\
$\WmvA,\WmvAi$ & Wallet states 
& $\gain{\AddrA}{\sysS}{\TxTS}$ & $\AddrA$'s gain upon firing $\TxTS$ in $\sysS$
\\
$\cstC,\cstD$ & Contract states 
& $\Adv$ & Set of adversaries
\\
\hline
\end{tabular}
\end{table}

We fix a model of account-based blockchains
and smart contracts \emph{\`a la} Ethereum.
To keep our theory simple, we abstract from
consensus features (\eg, the gas mechanism), and we
rule out some problematic behaviours (\eg, reentrancy).

\paragraph{Blockchain states}
We assume a set $\TokU$ of \emph{token types} ($\tokT, \tokTi, \ldots$)
and a countably infinite set $\AddrU$ of \emph{accounts}.
Accounts are partitioned into
\emph{user accounts} $\pmvA, \pmvB, \ldots \in \PmvU$
(representing the so-called \emph{externally owned accounts} in Ethereum)
and
\emph{contract accounts} $\cmvC, \cmvD, \ldots \in \CmvU$.
We designate a subset $\Adv$ of user accounts as adversaries%
\footnote{%
In practice, given a blockchain state it would be safe to
say that $\Adv$ are the accounts never mentioned in
the contract states and code.
Modelling $\Adv$ as a system parameter is a simplification,
which avoids to make our definitions depend on a specific contract language.
It would be possible to remove the parameter $\Adv$,
at the cost of an increased complexity of the definitions and statements
(see \eg the adversarial MEV in~\cite{BZ23mev}).},
\change{B7}{including \eg block proposers and MEV searchers}.
%
The state of a user account is a map \mbox{$\wmvA \in \TokU \rightarrow \Nat$}
from tokens to non-negative integers,
\ie a \keyterm{wallet} that quantifies the tokens in the account.
The state of a contract account is a pair $(\wmvA,\sigma)$,
where $\wmvA$ is a wallet and $\sigma$ is a key-value store.
\keyterm{Blockchain states} $\sysS, \sysSi, \ldots$
are finite maps from accounts to their states, 
where user wallets include at least $\Adv$'s.
We use the operator $\mid$ to deconstruct a blockchain state
into its components, writing \eg:
\[
\sysS =
\walpmv{\pmvA}{\waltok{1}{\tokT},\waltok{2}{\tokT[2]}} \mid
\walu{\pmvM}{0}{\tokT} \mid
\walpmv{\contract{C}}{\waltok{1}{\tokT},\code{owner}=\pmvA}
\]
for a blockchain state where the user account $\pmvA$
stores 1 unit of token $\tokT$ and two units of token $\tokT[2]$,
the user account $\pmvM$ has zero tokens,
and the contract $\contract{C}$ stores 1 unit of $\tokT$
and has a key-value store mapping $\code{owner}$ to $\pmvA$.

\paragraph{Contracts}
We do not rely on a specific contract language:
we just assume that contracts have an associated set of methods
(like \eg, in Solidity).%
\footnote{%
For simplicity, we forbid direct transfer of assets between users:
this is not a limitation, since these transfers can always be routed
by suitable contracts.
}
A method can:
\begin{inlinelist}
\item receive parameters and tokens from the caller,
\item update the contract wallet and state,  
\item transfer tokens to user accounts,
\item call other contracts (possibly transferring tokens along with the call),
\item return values and transfer tokens to the caller,
\item abort.
\end{inlinelist}
As usual, a method cannot drain tokens from other accounts:
the only ways for a contract to receive tokens are
\begin{inlinelist}
\item from a caller invoking one of its methods, or
\item by calling a method of another contract that sends tokens to its caller.
\end{inlinelist}
%
For simplicity,
we assume that a contract $\cmvC$ can only call methods of contracts
deployed before it.
Formally, defining ``\emph{$\cmvC$ is called by $\cmvD$}'' when some method of $\cmvD$
calls some method of $\cmvC$,
we are requiring that the transitive and reflexive closure $\sqsubseteq$
of this relation is a partial order.
We also assume that blockchain states contain all the \emph{dependencies}
of their contracts:
formally, if $\CmvC$ are the contracts in $\sysS$, we require that
$\deps{\CmvC} = \setcomp{\cmvCi}{\exists \cmvC \in \CmvC.\ \cmvCi \sqsubseteq \cmvC}$
are in $\sysS$.
States satisfying these assumptions are said \emph{well-formed}:
all states mentioned in our results (either in hypothesis or thesis)
are always well-formed.
Although well-formedness makes our model cleaner than Ethereum,
preventing some problematic behaviours like reentrant calls~\cite{Luu16ccs}, we only need it
on the analysed contracts $\cstD$ and their dependencies (see~\Cref{sec:related}).
We write \mbox{$\sysS = \WmvA \mid \cstC$}
for a blockchain state $\sysS$ composed 
of user wallets $\WmvA$ and contract states $\cstC$.
We can deconstruct wallets, writing \mbox{$\sysS = \WmvA \mid \WmvAi \mid \cstC$}
when $\dom{\WmvA}$ and $\dom{\WmvAi}$ are disjoint, as well as contract states, writing \mbox{$\sysS = \WmvA \mid \cstC \mid \cstD$}.
We denote by $\cmvOfcst{\cstC}$ the set of contract accounts
in $\cstC$ (\ie $\cmvOfcst{\cstC} = \dom{\cstC}$),
and let $\deps{\cstD} = \deps{\cmvOfcst{\cstD}}$.
Finally, we assume that contracts cannot inspect the state
of other accounts, including users' wallets
and the state of other contracts.%
\footnote{These are not restrictions in practice.
To make a contract depend on a user's wallet,
we can require the users to transfer tokens along with contract calls.
To make it depend on the state of other contracts,
we can access it through getter methods.}
Formally, we are requiring that each transaction
enabled in $\sysS$ produces the same effect
in a ``richer'' state $\sysSi \geq_{\$} \sysS$
containing more tokens in users' wallets
\iftoggle{arxiv}
{(\Cref{def:wallet-monotonicity}).}
{(Definition B.2 in~\cite{BMZ23defi}).}

\paragraph{Transactions}
We model contracts behaviour as a 
deterministic transition relation $\xrightarrow{}$ between blockchain states,
where state transitions are triggered by
\keyterm{transactions} $\txT, \txTi, \ldots$.
A transaction is a call to a contract method,
written $\pmvA:\contract{C}.\txcode{f}(\code{args})$,
where $\pmvA$ is the user signing the transaction,
$\contract{C}$ is the called contract,
$\txcode{f}$ is the called method,
and $\code{args}$ is the list of actual parameters,
which can also include transfers of tokens from $\pmvA$ to $\cmvC$.
Invalid transactions are rolled-back,
\ie $\xrightarrow{}$ preserves the state.
There are various reasons for invalidity,
\eg the called method aborts, a token transfer without the needed tokens
is attempted, \etc
Given $\txT = \pmvA:\contract{C}.\txcode{f}(\code{args})$,
we write $\callee{\txT}$ for the target contract~$\contract{C}$.
Methods can refer to $\pmvA$ via the identifier $\origin$
and to the caller (contract or user) account via $\sender$
(corresponding, resp., to $\code{tx.origin}$ and $\code{msg.sender}$ in Solidity).

\paragraph{Example: an Automated Market Maker}
In our examples, we specify contracts in pseudo-code.
Its syntax is similar to Solidity, with some extra features:
\begin{inlinelist}
\item the expression $\walE{\tokT}$ denotes the number of tokens $\tokT$
  stored in the contract;
\item the formal parameter $\txin{}{x}{\tokT}$
requires the $\sender$ to transfer some tokens $\tokT$
to the contract along with the call
(the unsigned integer variable $x$ generalises Solidity's \code{msg.value}
to multi-tokens);
\item the command $\txout{\addr{a}}{e}{\tokT}$ transfers
$e$ units of $\tokT$ from the contract to account $\addr{a}$,
where $e$ is an expression, and $\addr{a}$ could be either a user account
or the method $\sender$).
\end{inlinelist}
We exemplify pseudo-code in~\Cref{fig:amm},
specifying  an Automated Market Maker
inspired by Uniswap v2~\cite{Angeris20aft,Angeris21analysis,BCL22amm,Xu21sok}.
%
%
Users can add liquidity,
query the token pair and the exchange rate,
and swap units of $\tokT[0]$ with units of $\tokT[1]$
or vice-versa.
More specifically,
$\pmvA:\txcode{swap}(\txin{}{x}{\tokT[0]},y_{\min})$
allows $\pmvA$ to send $\waltok{x}{\tokT[0]}$ to the contract,
and receive at least \mbox{$\waltok{y_{\min}}{\tokT[1]}$} in exchange.
Symmetrically, $\pmvA:\txcode{swap}(\txin{}{x}{\tokT[1]},y_{\min})$
allows $\pmvA$ to exchange $\waltok{x}{\tokT[1]}$ for at least
\mbox{$\waltok{y_{\min}}{\tokT[0]}$}.

\begin{figure}[t]
\begin{lstlisting}[language=txscript,morekeywords={AMM,addLiq,swap,getTokens,getRate},classoffset=4,morekeywords={T0,T1,t},keywordstyle=\tokColor,frame=single]
contract (*$\contract{AMM}$*) {
  addLiq(?x0:T0,?x1:T1) { // add liquidity to the AMM
    require #T0 * (#T1-x1) == (#T0-x0) * #T1 }
  getTokens() { return (T0,T1) }    // token pair
  getRate(t) {                      // exchange rate
    if (t==T0) return #T0/#T1       // r:T0 for 1:T1
    else if (t==T1) return #T1/#T0  // r:T1 for 1:T0
    else abort }
  swap(?x:t,ymin) {
    if (t==T0)
      { y=(x*#T1)/#T0; require ymin<=y<#T1; sender!y:T1 }
    else if (t==T1)
      { y=(x*#T0)/#T1; require ymin<=y<#T0; sender!y:T0 }
    else abort }
}
\end{lstlisting}
\negcaptionspace
\caption{A constant-product AMM contract.}
\label{fig:amm}
\end{figure}

\hidden{
\begin{remark}
  \label{remark:well-formed-state}
  To improve the readability of our results,
  we adopt the following convention:
  every time the operator $\mid$ occurs in a statement,
  we implicitly assume in the statement hypotheses
  that the well-formedness condition above is respected.
\end{remark}
}

\paragraph{Wealth and gain}
Measuring the effect of an attack
requires to estimate the \keyterm{wealth} of the adversary
before and after the attack.
We denote by $\wealth{\AddrA}{\sysS}$
the wealth of accounts $\AddrA$ in $\sysS$.
Such wealth is given by the weighted sum of the tokens in $\AddrA$'s wallets,
where the weights are the token prices.
We denote by $\price{\tokT}$ the (strictly positive)
price of token type $\tokT$. 
\change{B5}{This implicitly assumes that token prices are constant, 
since they do not depend on the blockchain state.}%
\footnote{\changeNoMargin{This simplifying assumption allows local MEV to neglect the parts of the state that could affect token prices. A more realistic handling of token prices would require to extend the model with a function that determines the token prices in a given state.}}

\begin{definition}[Wealth]
  The wealth of $\AddrA \subseteq \AddrU$
  in $\sysS = \WmvA \mid \cstC$ is given by:
  \label{def:wealth}
  \begin{equation}
    \label{eq:wealth}
    \wealth{\AddrA}{\sysS}
    \; = \;
    \sum_{\pmvA \in \AddrA \cap \dom{\WmvA},\, \tokT}
    \hspace{-12pt}
    \WmvA(\pmvA)(\tokT) \cdot \price{\tokT}
    \hspace{4pt} +
    \sum_{\cmvC \in \AddrA \cap \dom{\cstC},\, \tokT}
    \hspace{-12pt}    
    \fst(\cstC(\cmvC))(\tokT) \cdot \price{\tokT}    
  \end{equation}
\end{definition}

To rule out ill-formed states with an infinite amount of tokens,
we require blockchain states to enjoy the \emph{finite tokens axiom}, \ie
$\sum_{\pmvA,\tokT} \sysS(\pmvA) (\tokT) \in \Nat$.
This makes the global wealth always finite.
The success of attacks is measured as \emph{gain},
\ie the difference of the attackers' wealth before and after the attack.

\begin{definition}[Gain]
  \label{def:gain}
  The gain of $\AddrA \subseteq \AddrU$ upon firing
  a transactions sequence $\TxTS$ in $\sysS$ is given by
  \(
  \gain{\AddrA}{\sysS}{\TxTS}
  =
  \wealth{\AddrA}{\sysSi} - \wealth{\AddrA}{\sysS}
  \)
  if $\sysS \xrightarrow{\TxTS} \sysSi$.
\end{definition}

\section{Local MEV}
\label{sec:mev}

The MEV in a state $\sysS$ is 
the maximum gain of the adversary $\Adv$ upon 
firing a sequence of transactions constructed
by $\Adv$~\cite{Babel23clockwork,BZ23mev}.
If we denote by $\mall{}{\Adv,P}$ the set of transactions craftable by $\Adv$ using a mempool $P$,
and by $\mall{}{\Adv,P}^*$ their finite sequences,
the MEV in a blockchain state $\sysS$ can be formalised as:
\begin{equation}
  \label{eq:mev}
  \mev{}{\sysS,P}{}
  = \max \setcomp
  {\gain{\Adv}{\sysS}{\TxTS}}
  {\TxTS \in \mall{}{\Adv,P}^*}
\end{equation}

This notion measures the value that adversaries can extract from \emph{any}
contract in the blockchain and transfer to their wallets.
To study composability,
following the intuitions in~\Cref{sec:intro}
we need instead to check if the contracts that will be deployed
can have a loss when adversaries manipulate their dependencies.
Therefore, our notion of MEV diverges from~\eqref{eq:mev} in four aspects:
\begin{enumerate}

\item We measure the value that adversaries can extract from
a \emph{given} set of contracts $\CmvC$.
This means that only the tokens extracted from the contracts in $\CmvC$
contribute to the extractable value,
while the tokens grabbed from other contracts do not count
(while they would count for the \emph{global} MEV in~\eqref{eq:mev}).

\item We count as MEV
  \emph{all} the tokens that $\Adv$ can remove from $\CmvC$,
  regardless of whether $\Adv$ can transfer them to their wallets.
  Namely, while~\eqref{eq:mev} maximises $\Adv$'s gain $\gain{\Adv}{}{}$,
  our notion maximises $\CmvC$'s \emph{loss} $-\gain{\CmvC}{}{}$.
  This is because we want our composability to provide security
  also against \emph{irrational adversaries},
  who try to damage contracts
  without necessarily making a profit.

\item 
  We parameterise our MEV \wrt the set $\CmvD$ of contracts callable by~$\Adv$.
  This allows us to rework the
  distinction between private and public information
  in language-based non-interference.
  There, the idea is that public outputs are not affected by private inputs,
  \ie restricting inputs to the private ones must preserve the public outputs.
  In our context,
  we rephrase this by requiring that the MEV extractable from $\CmvC$
  (the public output)
  is not affected when restricting the set of callable contracts
  to a subset $\CmvD$ (the private inputs).

\item \change{B3}{We assume that the mempool is empty,  
  just writing $\mall{}{\Adv}$. 
  We do so because to define secure composability 
  we are concerned about
  the MEV extractable by exploiting new contracts, and not that
  extractable from the mempool.
  Note that the mempool is instead considered in the MEV notions in~\cite{Babel23clockwork,BZ23mev}.}

\end{enumerate}
We call this new notion \keyterm{local MEV},
and we denote it by $\lmev{\CmvD}{\sysS}{\CmvC}$.

\begin{definition}[Local MEV]
  \label{def:lmev}
  Let
  \(
  \mall{\CmvD}{\Adv}
  = 
  \setcomp{\txT \in \mall{}{\Adv}}{\callee{\txT} \subseteq \CmvD}
  \)
  be the set of transactions craftable by $\Adv$ 
  and targeting contracts in $\CmvD$.
  We define:
  \begin{equation}
    \label{eq:lmev}
    \lmev{\CmvD}{\sysS}{\CmvC}
    = \max \setcomp
    {-\gain{\CmvC}{\sysS}{\TxTS}}
    {\TxTS \in \mall{\CmvD}{\Adv}^*}
  \end{equation}
  Hereafter, we abbreviate $\lmev{\CmvU}{\sysS}{\CmvC}$ as
  $\lmev{}{\sysS}{\CmvC}$.
\end{definition}

\begin{example}
  \label{ex:lmev:amm}
  Consider two instances of the AMM contract in~\Cref{fig:amm},
  to swap respectively the pairs $(\tokT[0],\tokT[1])$ and $(\tokT[1],\tokT[2])$.
  Let $\price{\tokT[0]} = \price{\tokT[1]} = \price{\tokT[2]} = 1$, and let
  \(
  \sysS = 
  \walpmv{\pmvM}{\waltok{3}{\tokT[0]}} \mid
  \walpmv{\contract{AMM1}}{\waltok{6}{\tokT[0]},\waltok{6}{\tokT[1]}} \mid
  \walpmv{\contract{AMM2}}{\waltok{4}{\tokT[1]},\waltok{9}{\tokT[2]}}
  \),
  where $\pmvM$ is the adversary.
  We want to compute the local MEV \wrt $\CmvC = \setenum{\contract{AMM2}}$.
  \change{B3}{Recall that we are assuming the mempool to be empty.}
  In the unrestricted case (\ie, $\CmvD$ is the universe),
  the trace $\TxTS$ that maximises the loss of $\contract{AMM2}$
  is the following:
  \begin{align*}
    \sysS
    & \xrightarrow{\pmvM:\contract{AMM1}.\txcode{swap}(\txin{}{3}{\tokT[0]},0)}
    && \walu{\pmvM}{2}{\tokT[1]} \mid
    \walpmv{\contract{AMM1}}{\waltok{9}{\tokT[0]},\waltok{4}{\tokT[1]}} \mid
    \walpmv{\contract{AMM2}}{\waltok{4}{\tokT[1]},\waltok{9}{\tokT[2]}}  
    \\
    & \xrightarrow{\pmvM:\contract{AMM2}.\txcode{swap}(\txin{}{2}{\tokT[1]},0)}
    && \walu{\pmvM}{3}{\tokT[2]} \mid
    \walpmv{\contract{AMM1}}{\waltok{9}{\tokT[0]},\waltok{4}{\tokT[1]}} \mid
    \walpmv{\contract{AMM2}}{\waltok{6}{\tokT[1]},\waltok{6}{\tokT[2]}}
  \end{align*}
  We have that $-\gain{\setenum{\contract{AMM2}}}{\sysS}{\TxTS} = 1$,
  hence 
  $\lmev{}{\sysS}{\setenum{\contract{AMM2}}} = 1$.
  Instead, when $\pmvM$ is restricted to use
  $\CmvD = \setenum{\contract{AMM2}}$,
  it has no way to obtain the tokens $\tokT[1]$ that are needed to
  extract value from $\contract{AMM2}$: therefore,
  $\lmev{\setenum{\contract{AMM2}}}{\sysS}{\setenum{\contract{AMM2}}} = 0$.
  \hfill\qedex
\end{example}

\Cref{lem:lmev} establishes some useful properties of local MEV.
\Cref{lem:lmev:mev} studies some border cases:
in particular, when the set of observed contracts $\CmvC$ is the universe,
local MEV over-approximates global MEV.
\Cref{lem:lmev:L-leq-H,lem:lmev:monotonicity}
state that widening the restricted contracts $\CmvD$ or the contract state
potentially increases the local MEV%
\footnote{%
Instead, $\CmvC \subseteq \CmvCi$ does not imply 
$\lmev{\CmvD}{\sysS}{\CmvC} \leq \lmev{\CmvD}{\sysS}{\CmvCi}$,
because maximising the loss of a contract may increase the gain of another one
\iftoggle{arxiv}
{(see~\Cref{ex:lmev:not-monotonic-on-observed-contracts}).}
{(see Example B.1 in~\cite{BMZ23defi}).}}.
\Cref{lem:lmev:garbage} allows to restrict the set of contract accounts
$\CmvC$ and $\CmvD$ to those occurring in the blockchain state $\cstC$ 
(\ie, $\cmvOfcst{\cstC}$).
\Cref{lem:lmev:leq-wealth} states that the
local MEV extractable from $\CmvC$ is non-negative
and bounded by the wealth of $\CmvC$.

\begin{lemma}[Basic properties of MEV]
  \label{lem:lmev}
  For all $\sysS$, $\CmvC,\CmvD \subseteq \CmvU$:
  \begin{enumerate}

  \item \label{lem:lmev:mev}
    $\lmev{\CmvD}{\sysS}{\emptyset} = \lmev{\emptyset}{\sysS}{\CmvC} = 0$,
    $\lmev{\CmvU}{\sysS}{\CmvU} \geq \mev{}{\sysS}{}$

  \item \label{lem:lmev:L-leq-H}
    if $\CmvD \subseteq \CmvDi$, then
    $\lmev{\CmvD}{\sysS}{\CmvC} \leq \lmev{\CmvDi}{\sysS}{\CmvC}$

  \item \label{lem:lmev:monotonicity}
    $\lmev{\CmvD}{\WmvA \mid \cstC}{\CmvC} \leq \lmev{\CmvD}{\WmvA \mid \cstCi}{\CmvC}$, where $\cstCi\mid_{\dom{\cstC}} = \cstC$

  \item \label{lem:lmev:garbage}
    $\lmev{\CmvD}{\WmvA \mid \cstC}{\CmvC} = \lmev{\CmvD}{\WmvA \mid \cstC}{\CmvC \cap \cmvOfcst{\cstC}} = \lmev{\CmvD\cap \cmvOfcst{\cstC}}{\WmvA \mid \cstC}{\CmvC}$
    
  \item \label{lem:lmev:leq-wealth}
    $0 \leq \lmev{\CmvD}{\sysS}{\CmvC} \leq \wealth{\CmvC}{\sysS}$
    
  \end{enumerate}
\end{lemma}

\Cref{lem:lmev-wallet} states that the only user wallets that
need to be taken into account to estimate the  MEV are those of the adversary
(\Cref{lem:lmev-wallet:zero-user}).
This is because $\Adv$ has no way to force other users to spend their
tokens in the attack sequence.%
\footnote{This property would not hold if the adversary could play
user transactions in the mempool, since their validity depends on
the wallets of the users who signed them.
However, when studying contract composability the question is whether
a new contract can be attacked by exploiting its dependencies, so
the mempool is irrelevant.}
Furthermore, wealthier adversaries may potentially extract more MEV 
(\Cref{lem:lmev-wallet:monotonicity}).

\begin{lemma}[MEV and adversaries' wallets]
  \label{lem:lmev-wallet}
  \begin{enumerate}

  \item \label{lem:lmev-wallet:zero-user}
    if $\dom{\WmvA[\Adv]} = \Adv$, then
    \(
    \lmev{\CmvD}{\WmvA[\Adv] \mid \WmvA \mid \cstC}{\CmvC} = \lmev{\CmvD}{\WmvA[\Adv] \mid \cstC}{\CmvC}
    \)
    
  \item \label{lem:lmev-wallet:monotonicity}
    if $\sysS \leq_{\$} \sysSi$, then
    $\lmev{\CmvD}{\sysS}{\CmvC} \leq \lmev{\CmvD}{\sysSi}{\CmvC}$
  \end{enumerate}
\end{lemma}

Since the overall amount of tokens held in
contract wallets is limited,
the MEV no longer increases when the adversary is rich enough
(\Cref{lem:lmev:stability}).
\change{A2}{Formally, we prove that there exists a threshold adversary wallet $\WmvA[\Adv]$ yielding the same MEV as any richer adversary wallet.
Together with~\cref{lem:lmev-wallet:monotonicity} of~\Cref{lem:lmev-wallet},
this ensures $\WmvA[\Adv]$ yields the maximum MEV over any adversary wallet.}

\begin{lemma}[Stability]
  \label{lem:lmev:stability}
  For all $\CmvC$, $\CmvD$, $\cstC$, 
  there exists an adversary wallet $\WmvA[\Adv]$ such that
  \(
  \lmev{\CmvD}{\WmvA[\Adv] \mid \cstC}{\CmvC}  
  =
  \lmev{\CmvD}{\WmvAi[\Adv] \mid \cstC}{\CmvC}    
  \)  
  for all $\WmvAi[\Adv] \geq_{\$} \WmvA[\Adv]$.
\end{lemma}

Taking the maximum MEV over all possible user wallets
(which always exists by~\Cref{lem:lmev:stability})
we then introduce a notion of MEV 
that does \emph{not} depend on user wallets at all (including $\Adv$'s).
The new notion, dubbed $\rlmev{}{}{}$, reflects the fact that
the capital required from the adversary is not a barrier in practice.
\Eg, as already noted in~\cite{Babel23clockwork},
the adversary can apply for a flash loan (a risk-free operation)
to obtain the capital needed to extract the full
\change{C1}{MEV}.%
\footnote{
\changeNoMargin{This is true for \emph{atomic} attacks, where the adversary extracts MEV through a single transaction. If extracting MEV requires multiple transactions, like \eg in sandwich attacks exploiting transactions in the mempool, flash loans are not possible.}}
Consequently, $\rlmev{}{}{}$ is 
the most robust estimation of the value extractable by the adversary.

\begin{definition}[Local MEV of wealthy adversaries]
  \label{def:rich-lmev}
  For all $\CmvC,\CmvD,\cstC$, let:
  \begin{equation}
    \label{eq:rich-lmev}
    \rlmev{\CmvD}{\cstC}{\CmvC}
    = \max_{\WmvA} \;
    \lmev{\CmvD}{\WmvA \mid \cstC}{\CmvC}
  \end{equation}
\end{definition}

Besides satisfying the basic properties of local MEV
seen before 
\iftoggle{arxiv}
{(\Cref{lem:rich-lmev})}
{(Lemma B.1 in~\cite{BMZ23defi})},
$\rlmev{}{}{}$ enjoys a key property:
to compute the $\rlmev{}{}{}$ of $\CmvC$ in~$\cstC$,
we can ignore all the contracts in $\cstC$ except the dependencies of $\CmvC$.
This is not true for $\lmev{}{}{}$, because
non-wealthy adversaries may need to
interact with other contracts to obtain the tokens needed to
extract value from $\CmvC$.
Formally, we define the \emph{stripping} of $\cstC$ \wrt $\CmvC$
(in symbols, $\strip{\cstC}{\CmvC}$)
as the restriction of $\cstC$ to the domain $\deps{\CmvC}$.
Note that when $\cstC$ is well-formed, then
also $\strip{\cstC}{\CmvC}$ is well-formed.
\Cref{th:rich-lmev:stripping} gives sufficient conditions under which
we can strip from $\cstC$ all the non-dependencies of $\CmvC$,
preserving $\rlmev{\CmvD}{\cstC}{\CmvC}$.
The first condition is that contract methods are not aware of
the identity of the $\sender$, being only able to use it
as a recipient of token transfers:
we refer to this by saying that contracts are \emph{sender-agnostic}
\iftoggle{arxiv}
{(see~\Cref{def:sender-agnostic}).}
{(see Def.~B.4 in~\cite{BMZ23defi}).}
The second condition ensures that $\CmvD$ contains enough contracts
to reproduce attacks in the stripped state.
 
\begin{theorem}[State stripping]
  \label{th:rich-lmev:stripping}
  \(
  \rlmev{\CmvD}{\cstC}{\CmvC} = \rlmev{\CmvD}{\strip{\cstC}{\CmvC}}{\CmvC}
  \)
  holds if the contracts
  $\CmvCi = \deps{\CmvC} \cap \deps{\CmvD \setminus \deps{\CmvC}}$ satisfy:
  \begin{inlinelist}

  \item \label{th:rich-lmev:stripping:1}
    $\CmvCi$ are sender-agnostic, and
    
  \item \label{th:rich-lmev:stripping:2}
    $\CmvCi \subseteq \CmvD$.
  \end{inlinelist}
  In particular, \ref{th:rich-lmev:stripping:2} holds if
  $\CmvD = \CmvU$ or $\CmvD = \CmvC$.
\end{theorem}

If any of the conditions \ref{th:rich-lmev:stripping:1}
or~\ref{th:rich-lmev:stripping:2} do not hold,
then adversaries might not be able to perform an attack on 
\mbox{$\strip{\cstC}{\CmvC}$} with the same effect as an attack on $\cstC$,
and so they might extract less
\mbox{$\rlmev{}{}{}$ in $\strip{\cstC}{\CmvC}$}
than in $\cstC$
\iftoggle{arxiv}
{(see~\Cref{ex:rich-lmev:stripping:1-false,ex:rich-lmev:stripping:2-false}).}
{(see Ex.~B.2, B.3 in~\cite{BMZ23defi}).}


The following~\namecref{th:rich-lmev:future-not-affect-past}
of~\Cref{th:rich-lmev:stripping}
states another key property of $\rlmev{}{}{}$:
extending the state with new contracts does not affect the MEV
extractable from the old contracts.%
\footnote{This property relies on assumptions
(\eg, non-circularity of contract dependencies)
that hold in our blockchain model, but may not hold in some
concrete platforms. For instance, in Ethereum one can craft contracts
that exploit reentrancy to extract MEV from already existing contracts,
as in the infamous DAO attack~\cite{DAO}.}
This is the basis to prove that $\richnonint{}{}$ is resistant
to adversarial contracts (\Cref{cor:richnonint:front-running}).

\begin{corollary}
  \label{th:rich-lmev:future-not-affect-past}
  \(
  \rlmev{}{\cstC}{\CmvC} = \rlmev{}{\cstC \mid \cstD}{\CmvC}
  \)
  for all $\CmvC \subseteq \cmvOfcst{\cstC}$.
\end{corollary}

\Cref{th:rich-lmev:future-not-affect-past} highlights a drawback
of the $\epsilon$-composability in~\cite{Babel23clockwork}, which
compares the global MEV in $\sysS$
with that in $\sysS \mid \cstD$,
where new contracts $\cstD$ have been deployed.
\Cref{th:rich-lmev:future-not-affect-past} states that
for wealthy adversaries,
the two MEVs may only differ in the value extractable from $\cstD$,
and so $\cstD$ is $0$-composable with $\sysS$ only if $\cstD$ has zero MEV.
Hence, contracts that have intended MEV
are not composable for~\cite{Babel23clockwork},
even if they have no interactions at all with the context.


\section{MEV non-interference}
\label{sec:non-interference}

To formalise  contract composability,
we start by defining a relation \mbox{$\nonint{\sysS}{\cstD}$}
between blockchain states
\mbox{$\sysS = \WmvA \mid \cstC$} and contract states $\cstD$.
Intuitively, \mbox{$\nonint{\sysS}{\cstD}$} means that
the adversary cannot leverage $\sysS$ to extract more MEV from $\cstD$
than it would be possible by interacting with $\cstD$ alone.
We will then say that
$\sysS$ is \emph{MEV non-interfering} with $\cstD$.
Note that $\nonint{\sysS}{\cstD}$ may or may not hold depending on the
wealth of the adversary in $\sysS$,
as we have already observed in~\Cref{lem:lmev-wallet}
that MEV depends on the adversary's wallets.
As noted before, assuming bounds on the capital available to the adversary
may be unsafe, and consequently we introduced in~\Cref{def:rich-lmev}
a notion of MEV that does not make assumptions on the adversary's wealth.
Based on this notion,
we will study later in this section another relation
$\richnonint{\cstC}{\cstD}$, which holds when contracts $\cstC$ do not interfere
with the MEV extractable from $\cstD$ \emph{regardless} of the
adversary's wealth.

Formally, MEV non-interference $\nonint{\sysS}{\cstD}$
holds when the MEV extractable from the contract accounts in $\cstD$
(\ie, $\cmvOfcst{\cstD}$) using \emph{any} contract in $\sysS \mid \cstD$
is exactly the same MEV that can be extracted using \emph{only}
the contracts in $\cstD$.
We write $\negnonint{\sysS}{\cstD}$ when $\nonint{\sysS}{\cstD}$ does not hold.%
\footnote{%
Note that these definitions only apply when $\sysS \mid \cstD$
is a well-formed blockchain state.}

\begin{definition}[MEV non-interference]
  \label{def:non-interference}
  A state $\sysS$ is \emph{$\lmev{}{}{}$ non-interfering} with $\cstD$,
  in symbols $\nonint{\sysS}{\cstD}$, when
  \(
  \lmev{}{\sysS \mid \cstD}{\cmvOfcst{\cstD}} = \lmev{\cmvOfcst{\cstD}}{\sysS \mid \cstD}{\cmvOfcst{\cstD}}
  \).%
  \footnote{
  Note that $\geq$ always holds by~\Cref{lem:lmev},
  so the definition only requires to check $\leq$. 
  }
\end{definition}

The following example discriminates MEV non-interference
from Babel \emph{et al.}' $\epsilon$-composability,
showing that a contract with intended MEV but no interactions with the
context enjoys MEV non-interference, but it is not $\epsilon$-composable.

\begin{figure}[t]
  \begin{lstlisting}[language=txscript,morekeywords={Airdrop,withdraw,Exchange,swap,getRate,setRate,getTokens},classoffset=4,morekeywords={a,b,A,Oracle},keywordstyle=\pmvColor,classoffset=4,morekeywords={t,t1,t2,T,tin,tout},keywordstyle=\tokColor,frame=single]
contract (*$\contract{Airdrop}$*) {
  constructor(?x:t) { tout=t }     // deposit any token t
  withdraw() { sender!#tout:tout } // any user withdraws
}
contract (*$\contract{Exchange}$*) {
  constructor(?x:t1,t2,r) {
    require r>0; rate=r; tout=t1; tin=t2; owner=origin }
  getTokens() { return (tin,tout) }
  getRate() { return rate }
  setRate(newRate) { require origin==owner; rate=newRate }  
  swap(?x:t) {               // receives x units of tin 
    require t==tin && #tout>=x*rate;
    sender!x*rate:tout }     // sends x*rate units of tout
}
  \end{lstlisting}
  \negcaptionspace
  \caption{An airdrop and an exchange contract.}
  \label{fig:airdrop}
  \label{fig:exchange}  
\end{figure}

\begin{example}
  \label{ex:nonint:airdrop}
  Consider the airdrop contract in~\Cref{fig:airdrop},
  let $\sysS$ be any blockchain state, and let
  $\cstD = \walpmv{\contract{Airdrop}}{\waltok{n}{\tokT},\code{tout}=\tokT}$.
  Babel \emph{et al.}' $\epsilon$-composability
  requires $\mev{}{\sysS \mid \cstD}{} \leq (1+\epsilon)\mev{}{\sysS}{}$.
  In \mbox{$\sysS \mid \cstD$}, the adversary can extract
  $n:\tokT$ from $\contract{Airdrop}$, and possibly use these tokens
  to extract more MEV from~$\sysS$.
  Therefore, 
  $\mev{}{\sysS \mid \cstD}{} \geq n \cdot \price{\tokT} + \mev{}{\sysS}{}$,
  meaning that, for large enough $n$,
  the airdrop is \emph{not} $\epsilon$-composable with $\sysS$.
  By contrast,
  \(
  \lmev{}{\sysS \mid \cstD}{\setenum{\contract{Airdrop}}}
  =
  n \cdot \price{\tokT}
  =
  \lmev{\setenum{\contract{Airdrop}}}{\sysS \mid \cstD}{\setenum{\contract{Airdrop}}}
  \),
  and so $\nonint{\sysS}{\cstD}$.
  This correctly reflects the fact that
  using the contracts in $\sysS$ does not give $\Adv$
  any way to damage the airdrop.
  \hfill\qedex
\end{example}

We now discuss contract conditions that may break MEV non-interference.
Clearly, \emph{contract dependencies}
(\ie, a contract in $\cstD$ that calls a contract in $\sysS$)
are a possible cause of MEV interferences $\negnonint{\sysS}{\cstD}$.
\Cref{ex:nonint:pricebet} shows the issue through a classical DeFi composition,
where a bet contract uses an AMM as a price oracle~\cite{Babel23clockwork}.
Another source of MEV interference is when the contracts
in $\sysS$ and $\cstD$ have \emph{token dependencies}
(\ie, a contract in $\cstD$ outputs tokens which can be used
as input to contract in $\sysS$, or vice-versa).
\Cref{ex:nonint:token-dependency} shows that this can cause
MEV interferences, even when $\sysS$ and $\cstD$ have no contract dependencies.


\begin{figure}[t]
  \begin{lstlisting}[language=txscript,morekeywords={bet,win,close,getTokens,getRate},classoffset=4,morekeywords={a,b,A},keywordstyle=\pmvColor,classoffset=5,morekeywords={t,ETH},keywordstyle=\tokColor,classoffset=6,morekeywords={Bet,oracle},keywordstyle=\cmvColor,frame=single]
contract (*$\contract[oracle]{Bet}$*) {
  constructor(?x:ETH,t,r,d) {
    require t!=ETH && oracle.getTokens()==(ETH,t);
    tok=t; rate=r; owner=origin; deadline=d }
  bet(?x:ETH) {
    require player==null && x==#ETH;
    player=origin }
  win() {
    require block.num<=deadline && origin==player;
    require oracle.getRate(ETH)>rate;
    player!#ETH:ETH }
  close() {
    require block.num>deadline && origin==owner;
    owner!#ETH:ETH }
}
  \end{lstlisting}
  \negcaptionspace
  \caption{A bet contract relying on an external price oracle.}
  \label{fig:pricebet}
\end{figure}

\begin{example}
  \label{ex:nonint:pricebet}
  The $\contract{Bet}$ contract in~\Cref{fig:pricebet}
  allows anyone to bet on the exchange rate between a token and $\ETH$.
  The contract is parameterised over an $\contract{oracle}$ contract
  that is queried for the token price.
  $\contract{Bet}$ receives the initial pot from the owner upon deployment.
  To join, a player must pay an amount of $\ETH$ equal to the pot.
  Before the deadline,
  the pot can be withdrawn by the player
  if the oracle exchange rate is greater than the bet rate.
  After the deadline, the pot goes to the owner.
  Consider an instance of $\contract{Bet}$ using an $\contract{AMM}$
  as price oracle, and let:
  \begin{align*}
  \sysS & = 
  \walu{\pmvM}{310}{\ETH} \mid
  \walpmv{\contract{AMM}}{\waltok{600}{\ETH},\waltok{600}{\tokT}} \mid
  \code{block.num}=n-k \mid
  \cdots
  \\
  \cstD & = \walpmv{\contract{Bet}}{\waltok{10}{\ETH},\code{tok}=\tokT,\code{rate}=3,\code{owner}=\pmvA,\code{deadline}=n}
  \end{align*}
  Note that the current oracle exchange rate is $1$,
  while winning the bet requires to make it exceed $3$.
  Since the deadline has not passed in $\sysS$
  ($\code{block.num}=n-k < n$) and the adversary $\pmvM$
  is rich enough, she can fire the following sequence:
  \begin{align*}
    \sysS \mid \cstD
    & \xrightarrow{\pmvM:\contract{Bet}.\txcode{bet}(\txin{}{10}{\ETH})}
    && \!\!\!\walu{\pmvM}{300}{\ETH} \mid
    \walpmv{\contract{AMM}}{\waltok{600}{\ETH},\waltok{600}{\tokT}} \mid
    \walpmv{\contract{Bet}}{\waltok{20}{\ETH},\cdots}
    \\
    & \xrightarrow{\pmvM:\contract{AMM}.\txcode{swap}(\txin{}{300}{\ETH},0)}
    && \!\!\!\walu{\pmvM}{200}{\tokT} \mid
    \walpmv{\contract{AMM}}{\waltok{900}{\ETH},\waltok{400}{\tokT}} \mid
    \walpmv{\contract{Bet}}{\waltok{20}{\ETH},\cdots}
    \\
    & \xrightarrow{\pmvM:\contract{Bet}.\txcode{win}()}
    && \!\!\!\walpmv{\pmvM}{\waltok{20}{\ETH},\waltok{200}{\tokT}} \mid
    \walpmv{\contract{AMM}}{\waltok{900}{\ETH},\waltok{400}{\tokT}} \mid
    \walpmv{\contract{Bet}}{\waltok{0}{\ETH},\cdots}
    \\
    & \xrightarrow{\pmvM:\contract{AMM}.\txcode{swap}(\txin{}{200}{\tokT},0)}
    && \!\!\!\walu{\pmvM}{320}{\ETH} \mid
    \walpmv{\contract{AMM}}{\waltok{600}{\ETH},\waltok{600}{\tokT}} \mid
    \walpmv{\contract{Bet}}{\waltok{0}{\ETH},\cdots}
  \end{align*}
  When $\pmvM$ can interact with $\contract{Bet}$'s dependency $\contract{AMM}$,
  the loss of $\contract{Bet}$ is $10:\ETH$, hence
  $\lmev{}{\sysS \mid \cstD}{\setenum{\contract{Bet}}} = 10 \cdot \price{\ETH}$.
  Instead, when $\pmvM$ can only use $\contract{Bet}$, 
  $\lmev{\setenum{\contract{Bet}}}{\sysS \mid \cstD}{\setenum{\contract{Bet}}} = 0$.
  Therefore,
  $\sysS$ is MEV interfering with the $\contract{Bet}$ contract.
  Note that a poorer $\pmvM$ may not have enough $\ETH$
  to produce the short-term volatility in the oracle exchange rate.
  Later in~\Cref{def:rich-non-interference} we will introduce a notion of
  MEV non-interference that does not depend on $\pmvM$'s wealth.
  \hfill\qedex
\end{example}

\begin{example}
  \label{ex:nonint:token-dependency}
  To show that MEV interference can happen even in the absence of
  contract dependencies,
  consider $\contract{Airdrop}$ and $\contract{Exchange}$
  in \Cref{fig:exchange}, 
  and let:
  \begin{align*}
    \sysS & =
    \walu{\pmvM}{0}{\tokT} \mid
    \walpmv{\contract{Airdrop}}{\waltok{1}{\tokT},\code{tout}=\tokT}
    \\
    \cstD & =
    \walpmv{\contract{Exchange}}{\waltok{10}{\ETH},\code{tin}=\tokT,\code{tout}=\ETH,\code{rate}=10,\code{owner}=\pmvB}
  \end{align*}
  The unrestricted MEV of $\contract{Exchange}$ is
  $10 \cdot \price{\ETH}$, since $\pmvM$ can first extract $1:\tokT$
  from the airdrop, and then use the exchange, draining $10 \cdot \price{\ETH}$.
  Instead, its restricted MEV is zero,
  since $\pmvM$ cannot obtain the needed $1:\tokT$.
  Hence, $\negnonint{\sysS}{\cstD}$.
  \hfill\qedex
\end{example}

\Cref{th:nonint:sufficient-conditions} devises sufficient conditions
for MEV non-interference.
Condition~\ref{th:nonint:sufficient-conditions:zero-mev} states that
\mbox{$\nonint{\sysS}{}{\cstD}$} holds whenever the new contracts $\cstD$
have zero MEV:
a special case is when $\cstD$ have no tokens,
by \Cref{lem:lmev}\eqref{lem:lmev:leq-wealth}.
Condition~\ref{th:nonint:sufficient-conditions:no-deps}
requires contract and token independence.
Formally, $\cstC$ and $\cstD$ are \emph{contract independent}
when their dependencies are disjoint, \ie
$\deps{\cstC} \cap \deps{\cstD} = \emptyset$,
and they are \emph{token independent} in $\sysS$
when the token types that can be received by $\cstC$ in $\sysS$
are disjoint from those that can be sent by $\cstD$ in $\sysS$,
and vice-versa
\iftoggle{arxiv}
{(see~\Cref{def:token-independence}).}
{(see Definition B.5 in~\cite{BMZ23defi}.)}
For instance, in~\Cref{ex:nonint:token-dependency}
the $\contract{Airdrop}$ has no input tokens,
and it has $\tokT$ as output token,
while $\contract{Exchange}$ has $\tokT$ as input token and
$\ETH$ as output token.
Since $\tokT$ is both in the outputs of $\contract{Airdrop}$ and in
the inputs of $\contract{Exchange}$, the two contracts are \emph{not}
token independent.
%
Condition~\ref{th:nonint:sufficient-conditions:adv-inert}
relaxes condition~\ref{th:nonint:sufficient-conditions:no-deps},
allowing contract dependence provided that
the dependencies of $\cstD$ occurring in $\cstC$
cannot be exploited by adversaries.
Formally, we require that $\cstD$ is \emph{stable \wrt moves of $\Adv$ on $\cstC$}, \ie any transaction craftable by $\Adv$ and targeting a contract in $\cstC$ does not affect the observable behaviour of methods of $\cstC$ called from $\cstD$.
The observable behaviour includes the returned values, the transferred tokens and the aborts of methods that can be called from $\cstD$. 



\begin{theorem}[Sufficient conditions for $\nonint{}{}$]
  \label{th:nonint:sufficient-conditions}
  Let \mbox{$\sysS = \WmvA \mid \cstC$}.
  Each of the following conditions implies $\nonint{\sysS}{\cstD}$:
  \begin{inlinelist}[(1)]

  \item \label{th:nonint:sufficient-conditions:zero-mev}
    $\lmev{}{\sysS \mid \cstD}{\cmvOfcst{\cstD}} = 0$

  \item \label{th:nonint:sufficient-conditions:no-deps}
    $\cstC$ and $\cstD$ are token independent in \mbox{$\sysS \mid \cstD$}
    and contract independent

  \item \label{th:nonint:sufficient-conditions:adv-inert}
    $\cstC$ and $\cstD$ are token independent in \mbox{$\sysS \mid \cstD$} and $\cstD$ is stable \wrt moves of $\Adv$ on $\cstC$.
  \end{inlinelist}
\end{theorem}

Note that neither contract independence nor token independence
nor stability
are necessary conditions for $\nonintrel$.
For instance, a contract that provides the best swap between two AMMs
(see~\iftoggle{arxiv}{\Cref{fig:bestswap}}{Figure A.1 in~\cite{BMZ23defi}})
has both contract dependencies
(on the two called AMMs) and token dependencies (their underlying tokens),
but nonetheless it enjoys MEV non-interference, because it has zero MEV
(indeed, the contract balance is always zero, so there is nothing to extract)%
\footnote{
By contrast, zero-MEV is \emph{not} a sufficient condition for
Babel \emph{et al.}' $\epsilon$-composability:
\eg, a contract may not be $0$-composable in $\sysS$
when $\Adv$ does not have the tokens needed to extract MEV
in $\sysS$, but becomes able to do so
after exchanging tokens between $\sysS$ and $\cstD$, provided that
this preserves $\cstD$'s wealth
\iftoggle{arxiv}
{(see~\Cref{cex:th:nonint:sufficient-conditions:zero-mev:babel}).}
{(see Example B.4 in~\cite{BMZ23defi}).}
}.


\paragraph{MEV non-interference against wealthy adversaries}

In general,
\mbox{$\nonint{\sysS}{\cstD}$} does not imply
\mbox{$\nonint{\sysS \mid \cstC}{\cstD}$},
even if $\cstC$ and $\cstD$ are contract independent.
This means that the notion studied so far
(similarly to $\epsilon$-composability of~\cite{Babel23clockwork})
gives no security guarantees against attacks
where adversaries manage to deploy contracts $\cstC$ \emph{before} $\cstD$.
For instance,
consider the contract
\mbox{$\cstC = \walpmv{\contract{Airdrop}}{\waltok{1}{\tokT},\code{tout}=\tokT}$} and the contract
\mbox{$\cstD = \walpmv{\contract[AMM]{Bet}}{\waltok{10}{\ETH},\code{tok}=\tokT,\code{rate}=3,\code{owner}=\pmvA,\code{deadline}=n}$},
and let $\sysS$ be empty.
Since $\pmvM$ has $0$ tokens in $\sysS$,
she cannot extract MEV from $\cstD$, hence $\nonint{\sysS}{}{\cstD}$
by condition~\ref{th:nonint:sufficient-conditions:zero-mev} of
\Cref{th:nonint:sufficient-conditions}.
Instead, $\negnonint{\sysS \mid \cstC}{}{\cstD}$,
as shown in~\Cref{ex:nonint:token-dependency}.
This highlights a usability limitation of $\nonint{}{}$:
assume that a user detects that $\nonint{\sysS}{\cstD}$,
and then sends a transaction to deploy the contracts $\cstD$.
If this transaction is front-run with another transaction
that deploys $\cstC$, then the new state $\sysS \mid \cstC$
could violate MEV non-interference with $\cstD$.
To overcome this limitation, \Cref{def:rich-non-interference}
provides a notion of MEV non-interference that is robust 
\wrt the adversary's wealth and enjoys resistance
to adversarial contracts.

\begin{definition}[MEV non-interference against wealthy adversaries]
  \label{def:rich-non-interference}
  A contract state $\cstC$ is \emph{$\rlmev{}{}{}$ non-interfering} with $\cstD$,
  in symbols $\richnonint{\cstC}{\cstD}$, when
  \[
  \rlmev{}{\cstC \mid \cstD}{\cmvOfcst{\cstD}}
  =
  \rlmev{\cmvOfcst{\cstD}}{\cstC \mid \cstD}{\cmvOfcst{\cstD}}
  \]
\end{definition}

\Cref{lem:richnonint-nonint} characterizes
$\richnonint{}{}$ in terms of $\nonint{}{}$:
intuitively, $\richnonint{\cstC}{\cstD}$ holds whenever
$\nonint{\WmvA \mid \cstC}{\cstD}$ holds for rich enough adversaries'
wallets $\WmvA$.%

\begin{lemma}[$\richnonint{}{}$ vs. $\nonint{}{}$]
  \label{lem:richnonint-nonint}
  $\richnonint{\cstC}{\cstD}$ if and only if
  $\exists \WmvA[0] . \ \forall \WmvA \geq_{\$} \WmvA[0] . \ \nonint{\WmvA \mid \cstC}{\cstD}$
\end{lemma}


The following~\namecref{th:richnonint:sufficient-conditions}
refines \Cref{th:nonint:sufficient-conditions}
giving sufficient conditions for $\richnonint{}{}$.
Note that token independence is no longer required.

\begin{theorem}[Sufficient conditions for $\richnonint{}{}$]
  \label{th:richnonint:sufficient-conditions}
  Each of the following conditions implies
  \mbox{$\richnonint{\cstC}{\cstD}$}: \
  \begin{inlinelist}[(1)]

  \item \label{th:richnonint:sufficient-conditions:zero-mev}
    $\rlmev{}{\cstC \mid \cstD}{\cmvOfcst{\cstD}} = 0$;

  \item \label{th:richnonint:sufficient-conditions:no-deps}
    $\cstC$ and $\cstD$ are contract independent;

  \item \label{th:richnonint:sufficient-conditions:adv-inert}
    $\cstD$ is stable \wrt moves of $\Adv$ on $\cstC$.
    
  \end{inlinelist}
\end{theorem}

\Cref{th:richnonint:stripping} establishes a key property
of $\richnonint{}{}$:
namely, \mbox{$\richnonint{\cstC}{\cstD}$} is preserved 
when removing from $\cstC$ all the contracts except
the dependencies of $\cstD$.
Formally. the statement filters $\cstC$ with the state stripping operator
$\strip{}{\cmvOfcst{\cstD}}$, 
which preserves exactly the dependencies of $\cstD$.
This highlights the algorithmic advantage of $\richnonint{}{}$
\wrt $\epsilon$-composability, which requires to consider the
whole blockchain state.

\begin{theorem}[State stripping]
  \label{th:richnonint:stripping}
  $\richnonint{\cstC}{\cstD}$
  if and only if
  $\richnonint{\strip{\cstC}{\cmvOfcst{\cstD}}}{\cstD}$
  when the contracts in
  $\deps{\cstD} \cap \deps{\cmvOfcst{\cstC}\setminus\deps{\cstD}}$
  are sender-agnostic.
\end{theorem}

\change{B4,C2}{%
A direct consequence of~\Cref{th:richnonint:stripping},
when the dependencies of $\cstD$ are sender-agnostic,
is that the adversary attacking $\cstD$ gains no advantage from 
creating additional contracts $\cstC[\Adv]$ before the attack.
Indeed, if $\richnonint{\cstC}{\cstD}$ holds,
then it also holds when the starting state is extended with
$\cstC[\Adv]$.
Intuitively, this is because any MEV extraction exploiting
$\cstC[\Adv]$ can also be performed by directly calling the
dependencies of $\cstD$, since they are not influenced
by the caller identity.%
}%
\begin{corollary}[\change{B4}{Resistance to adversarial contracts}]
  \label{cor:richnonint:front-running}
  If $\richnonint{\cstC}{\cstD}$, then 
  $\richnonint{\cstC \mid \cstC[\Adv]}{\cstD}$
  when the contracts in $\deps{\cstD}$ are sender-agnostic.
\end{corollary}

\begin{table}[t]
  \caption{Structural properties of $\richnonintrel$.
  Starred properties additionally require that some
  contracts are sender-agnostic.
  \iftoggle{arxiv}{}{Ref.\ lemmas/examples are in~\cite{BMZ23defi}.}
  }
  \label{tab:richnonint:struct-properties}
  \setlength{\tabcolsep}{7pt} 
  \renewcommand{\arraystretch}{1} 
  \centering
  \begin{tabular}{|c|c|c|c|}
    \hline
    \textbf{Hypothesis} & \textbf{Thesis} & \textbf{Valid} & \textbf{Ref.} \\
    \hline
    \multirow{4}{*}{$\richnonint{\cstC}{\cstD}$}
    & $\richnonint{\cstC \mid \cstCi}{\cstD}$
    & \cmark${}^{\star}$
    & Cor.~\ref{cor:richnonint:front-running}
    \\    
    & $\richnonint{\cstCi \mid \cstC}{\cstD}$
    & \cmark
    & \iftoggle{arxiv}{Lem.~\ref{lem:richnonint:mid-L}}{Lem.~B.2}
    \\
    & $\richnonint{\cstC}{\cstD \mid \cstDi}$
    & \xmark
    & \iftoggle{arxiv}{Ex.~\ref{ex:richnonint:mid-R}}{Ex.~B.6}
    \\    
    & $\richnonint{\cstC}{\cstDi \mid \cstD}$
    & \xmark
    & \iftoggle{arxiv}{Ex.~\ref{ex:richnonint:mid-R}}{Ex.~B.6}
    \\
    \hline
    $\richnonint{\cstC \mid \cstCi}{\cstD}$
    & \multirow{4}{*}{$\richnonint{\cstC}{\cstD}$}
    & \cmark${}$
    & \iftoggle{arxiv}{Lem.~\ref{lem:richnonint:erasure}}{Lem.~B.3}
    \\
    $\richnonint{\cstCi \mid \cstC}{\cstD}$
    & 
    & \cmark${}$
    & \iftoggle{arxiv}{Lem.~\ref{lem:richnonint:erasure}}{Lem.~B.3}
    \\
    $\richnonint{\cstC}{\cstD \mid \cstDi}$
    & 
    & \xmark
    & \iftoggle{arxiv}{Ex.~\ref{ex:richnonint:mid-L}}{Ex.~B.7}
    \\
    $\richnonint{\cstC}{\cstDi \mid \cstD}$
    & 
    & \xmark
    & \iftoggle{arxiv}{Ex.~\ref{ex:richnonint:mid-L}}{Ex.~B.7}
    \\
    \hline
    $\richnonint{\cstC}{\cstD[1]}$ $\,\land\,$
    $\richnonint{\cstC}{\cstD[2]}$
    & \multirow{3}{*}{$\richnonint{\cstC}{\cstD[1] \mid \cstD[2]}$}
    & \xmark
    & \iftoggle{arxiv}{Ex.~\ref{ex:richnonint:union}}{Ex.~B.8}
    \\
    $\richnonint{\cstC}{\cstD[1]}$ $\,\land\,$
    $\richnonint{\cstC \mid \cstD[1]}{\cstD[2]}$
    & 
    & \xmark
    & \iftoggle{arxiv}{Ex.~\ref{ex:richnonint:union}}{Ex.~B.8}
    \\    
    $\richnonint{\cstC}{\cstD[1]} \,\land\, \rlmev{}{\cstC \mid \cstD[1] \mid \cstD[2]}{\cmvOfcst{\cstD[2]}} = 0$
    & 
    & \cmark${}^{\star}$
    & \iftoggle{arxiv}{Lem.~\ref{lem:richnonint:mid-L-zero-mev}}{Lem.~B.4}
    \\    
    \hline
  \end{tabular}
\end{table}

We summarize in~\Cref{tab:richnonint:struct-properties}
some structural properties of $\richnonint{}{}$.
The first block shows that it is possible to extend
the LHS of $\richnonint{\cstC}{\cstD}$ with new contracts,
while in general it is not possible to extend the RHS.
The second block shows that
it is possible to cut contracts from the LHS, but not from the RHS.
The last block studies how to securely deploy a compound
contract \mbox{$\cstD[1] \mid \cstD[2]$} in a state $\cstC$.
To have $\richnonint{\cstC}{}{\cstD[1] \mid \cstD[2]}$ it is not
enough to independently check $\richnonint{\cstC}{}{\cstD[1]}$ and
$\richnonint{\cstC}{}{\cstD[2]}$.
Surprisingly, even after checking the secure deployment
of the first contract (\mbox{$\richnonint{\cstC}{}{\cstD[1]}$}) and then
that of the second contract in the resulting state
(\mbox{$\richnonint{\cstC \mid \cstD[1]}{}{\cstD[2]}$}),
we cannot guarantee the MEV non-interference of \mbox{$\cstD[1] \mid \cstD[2]$}.
The last row gives a sufficient condition
when the component $\cstD[2]$ has zero MEV.

\paragraph{DeFi compositions}

We now evaluate MEV non-interference of typical DeFi compositions.
Each line in~\Cref{tab:defi-compositions} shows
a compound contract $\cstD$, the context $\cstC$, 
and whether $\richnonintrel$ holds (\cmark) or not (\xmark).
Rows marked \cmark\ follow by~\Cref{th:richnonint:sufficient-conditions},
and indicate which case of the theorem applies.
We omit the contract states in $\cstC$ and $\cstD$:
a row marked \cmark\ means that $\richnonint{}{}$
holds for \emph{all} contract states,
while one marked \xmark\ means that $\richnonint{}{}$ fails to hold
for \emph{some} state.
Note that the results still hold for larger $\cstC$,
by~\Cref{th:richnonint:stripping}.
The pseudo-code of contracts is 
\iftoggle{arxiv}
{in~\Cref{sec:defi-compositions}.}
{\cite{BMZ23defi}.}

The first row shows that an AMM is always MEV non-interfering with another AMM.
This holds because $\richnonint{}{}$ assumes wealthy adversaries,
who always have enough tokens to attack the new AMM in $\cstD$,
without the need of exploiting the context $\cstC$.
Note that poor adversaries
(like in $\nonint{}{}$ and in $\epsilon$-composability)
could instead need to extract tokens from $\cstC$
in order to attack $\cstD$.
The second and third rows analyse the $\contract{Bet}$ contract.
We know from~\Cref{ex:nonint:pricebet} that $\contract{Bet}$
is not MEV non-interfering with an $\contract{AMM}$, since the adversary can
always win the bet by creating a price volatility in the $\contract{AMM}$.
This also holds for $\epsilon$-composability.
As expected, MEV non-interference holds when using $\contract{Exchange}$
as a price oracle, since the adversary cannot affect the exchange rate.
Note that $\epsilon$-composability does not properly capture the security of this composition:
indeed, $\contract{Bet}$ and $\contract{Exchange}$ are not $0$-composable
when the adversary (honestly) wins the bet,
since winning the bet is indistinguishable from extracting MEV.
We then consider two DeFi contracts that act as wrappers of AMMs:
$\contract{BestSwap}$ allows users to
perform the most profitable swap between two AMMs,
while $\contract{SwapRouter}$ routes a swap of $(\tokT[0],\tokT[2])$ across two AMMs
for $(\tokT[0],\tokT[1])$ and $(\tokT[1],\tokT[2])$.
Despite both contracts have AMMs in their dependencies,
they are always securely composable
(both $\richnonintrel$ and $0$-composable),
because they have zero balance across calls.
The same holds for $\contract{BestSwapRouter}$,
that provides the most profitable swap between two
$\contract{SwapRouter}$s.
The contracts $\contract{LPArbitrage}$ and $\contract{FlashLoanArbitrage}$
perform a risk-free arbitrage between two AMMs,
atomically borrowing and repaying a loan from a Lending Pool (LP)
(in $\contract{FlashLoanArbitrage}$, with no collateral).
As before, $\richnonintrel$ follows since the contracts
do not hold a balance between calls.
Instead, \mbox{$0$-composability} does not hold when the LP fee
is greater than $0$.
If so, borrowing increases the LP balance,
possibly increasing the MEV opportunities.

\begin{table}[t!]
  \caption{MEV non-interference of common DeFi compositions
  \iftoggle{arxiv}
  {(see~\Cref{sec:defi-compositions}).}
  {(\cite{BMZ23defi}, App.~A).}
  }
  \label{tab:defi-compositions}
  \setlength{\tabcolsep}{7pt} 
  \renewcommand{\arraystretch}{1.05} 
  \centering
  \begin{tabular}{|c|c|c|}
  \hline
  \textbf{Dependencies $\cstC$} & \textbf{New contracts $\cstD$} & \textbf{$\richnonint{\cstC}{\cstD}$}
  \\
  \hline
  $\contract{AMM}$
  & $\contract{AMM}$
  & \cmark${}^{\ref{th:richnonint:sufficient-conditions:no-deps}}$
  \\
  $\contract{AMM}$
  & $\contract[AMM]{Bet}$
  & \xmark
  \\
  $\contract{Exchange}$
  & $\contract[Exchange]{Bet}$
  & \cmark${}^{\ref{th:richnonint:sufficient-conditions:adv-inert}}$
  \\
  $\contract{AMM1} \mid \contract{AMM2}$
  & $\contract[AMM1,AMM2]{BestSwap}$
  & \cmark${}^{\ref{th:richnonint:sufficient-conditions:zero-mev}}$
  \\
  $\contract{AMM1} \mid \contract{AMM2}$
  & $\contract[AMM1,AMM2]{SwapRouter}$
  & \cmark${}^{\ref{th:richnonint:sufficient-conditions:zero-mev}}$
  \\
  $\contract{AMM1} \mid \contract{AMM2} \mid \contract[AMM1,AMM2]{SwapRouter1} \mid$
  & \multirow{2}{*}{$\contract[SwapRouter1,SwapRouter2]{BestSwap}$}
  & \multirow{2}{*}{\cmark${}^{\ref{th:richnonint:sufficient-conditions:zero-mev}}$}
  \\ $\contract{AMM3} \mid \contract{AMM4} \mid \contract[AMM3,AMM4]{SwapRouter2} \;\,$
  & 
  & 
  \\
  $\contract{AMM1} \mid \contract{AMM2} \mid \contract{LP}$
  & $\contract[AMM1,AMM2,LP]{LPArbitrage}$
  & \cmark${}^{\ref{th:richnonint:sufficient-conditions:zero-mev}}$
  \\
  $\contract{AMM1} \mid \contract{AMM2} \mid \contract{LP}$
  & $\contract[AMM1,AMM2,LP]{FlashLoanArbitrage}$
  & \cmark${}^{\ref{th:richnonint:sufficient-conditions:zero-mev}}$
  \\
  \hline
  \end{tabular}
\end{table}

\section{Discussion}
\label{sec:related}

We have proposed MEV non-interference,
a new security notion for DeFi composition
which ensures that adversaries cannot inflict economic harm
on compound contracts by exploiting their dependencies.
We have shown that our notion overcomes the drawbacks of
$\epsilon$-composability, the only other related notion
in literature~\cite{Babel23clockwork}.
In particular, while $\epsilon$-composability is a
property of the \emph{whole} blockchain state, 
MEV non-interference only needs to inspect
the newly deployed contracts and their dependencies
(\Cref{th:richnonint:stripping}).
\change{B1}{
The two notions are incomparable.
We already know from~\Cref{ex:nonint:airdrop}
that MEV non-interference does not imply \mbox{$\epsilon$-composability}.
\Cref{ex:babel-composability-not-implies-mev-nonintererence} below
shows the converse non-implication.
}

\begin{example}
\change{B1}{
  \label{ex:babel-composability-not-implies-mev-nonintererence}
  Babel \emph{et al.}' composability does not imply MEV non-interference.
  To show this, we craft $0$-composable
  $\cstC$ and $\cstD$ where both contracts expose
  identical and \emph{mutually exclusive} MEV:
  one can extract MEV from either contract, but not from both. 
  The choice of extracting MEV from $\cstC$ or from $\cstD$ 
  is done by calling a suitable method of $\cstC$.
  Only after the choice is made the MEV is exposed, 
  and the MEV from the other contract is permanently disabled. 
  %
  Since the two MEVs are mutually exclusive, we obtain 
  $\mev{}{\cdots \mid \cstC}{} = \mev{}{\cdots \mid \cstC \mid \cstD}{}$, hence \mbox{$0$-composability}.
  Non-interference fails because extracting MEV
  from $\cstD$ requires calling a method of $\cstC$.
  More concretely, using the contracts in \Cref{fig:babel-notimply-nonint}, let
  \(
  \cstC = 
  \walpmv{\contract{C1}}{\waltok{1}{\tokT},\code{n}=0}
  \),
  and let
  \(
  \sysS = \walu{\pmvM}{0}{\tokT} \mid \cstC
  \). 
  Assume that $\price{\tokT} = 1$.
  Let
  \(
  \cstD = \walu{\contract{C2}}{1}{\tokT}
  \).
  To study Babel \emph{et al.}' composability,
  we compare the (global) MEV of $\sysS$ and $\sysS \mid \cstD$.
  We have that
  $\mev{}{\sysS}{} = 1$
  by the sequence $\pmvM:\contract{C1}.\txcode{f1}()$,
  and $\mev{}{\sysS \mid \cstD}{} = 1$
  by the sequence \mbox{$\pmvM:\contract{C1}.\txcode{f2}() \ \pmvM:\contract{C2}.\txcode{g}()$}
  (or, alternatively, by the sequence $\pmvM:\contract{C1}.\txcode{f1}()$, which provides the same MEV).
  Therefore, the two contracts are $0$-composable.
  To study MEV non-interference, we have that
  $\lmev{}{\sysS \mid \cstD}{\setenum{\contract{C2}}} = 1$
  by the sequence
  \mbox{$\pmvM:\contract{C1}.\txcode{f2}() \ \pmvM:\contract{C2}.\txcode{g}()$},
  while
  $\lmev{\setenum{\contract{C2}}}{\sysS \mid \cstD}{\setenum{\contract{C2}}} = 0$,
  since the only callable method, \ie $\txcode{g}()$, always fails.
  Therefore, $\negnonint{\sysS}{}{\cstD}$,
  \ie $\cstD$ is not composable with $\sysS$ according to our notion.
  \hfill\qedex
}
\end{example}

\begin{figure}[t]
  \begin{lstlisting}[language=txscript,morekeywords={f1,f2,f3,g},classoffset=4,morekeywords={a,b,A,Oracle},keywordstyle=\pmvColor,classoffset=5,morekeywords={t,T,T2},keywordstyle=\tokColor,classoffset=6,morekeywords={C1,C2,C3},keywordstyle=\cmvColor,frame=single]
contract C1 {
  constructor(?x:T) { require x==1; n=0 }
  f1() { require (n==0); n=1; sender!1:T }
  f2() { require (n==0); n=2 }
  f3() { return n }
}
contract C2 {
  constructor(?x:T) { require x==1 }    
  g() { require (C1.f3()==2); sender!1:T }
}
  \end{lstlisting}
  \negcaptionspace
  \caption{Babel \emph{et al.}' composability does not imply MEV non-interference.}
  \label{fig:babel-notimply-nonint}
\end{figure}

We have studied sufficient conditions for MEV non-interference,
and rules to enable its modular verification
(\Cref{th:richnonint:sufficient-conditions},
\Cref{tab:richnonint:struct-properties}).
We have shown that these rules allow to correctly classify
the composability of common DeFi protocols (\Cref{tab:defi-compositions}).
As future work, we envision MEV non-interference 
as the basis of \emph{quantitative} versions of DeFi composability,
which give upper bounds to the loss of a compound contract
caused by manipulation of its dependencies.

\change{B2}{
A relevant question is whether simpler notions of composability
would achieve the same effect as our~\Cref{def:non-interference}.
For instance, one might be tempted to regard two contracts
$\cmvC$ and $\cmvCi$ composable whenever the (global) MEV of their
composition is equal to the sum of the two individual MEVs,
thus obtaining a property which could be written along the line of
\begin{equation}
\label{eq:wrong-composability:1}
\lmev{}{\walpmv{\contract{C}}{\cdot} \mid \walpmv{\contract{C}'}{\cdot}}{}
= \lmev{}{\walpmv{\contract{C}}{\cdot}}{}
+ \lmev{}{\walpmv{\contract{C}'}{\cdot}}{}
\end{equation}
While temptingly simple, this notion has several issues.
First, it does not consider the dependencies of the two contracts,
which must be part of the blockchain state.
\Eg, when both contracts depend on a third contract $\cmvD$, we could
amend~\eqref{eq:wrong-composability:1} as:
\begin{equation}
\label{eq:wrong-composability:2}
\lmev{}{\walpmv{\contract{D}}{\cdot} \mid \walpmv{\contract{C}}{\cdot} \mid \walpmv{\contract{C}'}{\cdot}}{}
= \lmev{}{\walpmv{\contract{D}}{\cdot} \mid \walpmv{\contract{C}}{\cdot}}{}
+ \lmev{}{\walpmv{\contract{D}}{\cdot} \mid \walpmv{\contract{C}'}{\cdot}}{}
\end{equation}
However, this equation is almost always false, since in the RHS 
the MEV extractable from $\cmvD$ is counted twice.
Even when $\cmvD$ has no MEV, it can still act as a shared state between 
$\cmvC$ and $\cmvCi$, \eg making it possible to extract MEV from only one of them (but not both). 
This, again, can be used to falsify~\eqref{eq:wrong-composability:2}. 
%
Furthermore, \eqref{eq:wrong-composability:2} does
not mention the wallet of the adversary.
Using the same wallet in each $\lmev{}{\cdot}{}$ would duplicate the adversary wealth in the RHS, leading to similar double-counting issues as those discussed previously for $\cmvD$.
}

We discuss some limitations of our work.
First, our blockchain model simplifies Ethereum
by requiring that contract dependencies are statically known and acyclic.
It looks feasible to refine our results to weaken the assumptions,
so that they are only required on the contracts $\cstD$ tested for composability
(see the end of 
\iftoggle{arxiv}
{\Cref{sec:proofs}}
{Appendix B in~\cite{BMZ23defi}} 
for a detailed discussion).
Note that with this refinement,
the rest of the blockchain state is not constrained,
making our results applicable to a wider class of contracts.
%
Another limitation is that local MEV measures the loss of a contract
as the value of the tokens that adversaries can remove from it:
in practice, adversaries could harm contracts also by \emph{freezing}
tokens without actually extracting them,
as in the infamous Parity Wallet attack~\cite{parity17nov}.
Refining local MEV to take these attacks into account could be done
by adapting the notion of liquidity~\cite{BLMZ22lmcs}.
A further possible improvement of our results is weakening the
sufficient conditions of~\Cref{th:richnonint:sufficient-conditions}:
in particular, condition~\ref{th:richnonint:sufficient-conditions:adv-inert}
is unnecessarily strict, since it forbids benign alterations of the
state, which do not affect the loss of $\cstD$.
Standard static analysis techniques for information flow~\cite{Volpano96jcs,Dimitrova12vmcai,Bossi07jcs}
could be adapted to refine this condition.
Finally, to keep our model simple we assumed that
the price of tokens is
constant (see~\Cref{sec:blockchain}) and that the
mempool is empty (\Cref{sec:mev}). Relaxing these assumptions
is left as future work.

\paragraph*{Acknowledgments}

This work was partially supported by project SERICS (PE00000014)
under the MUR National Recovery and Resilience Plan funded by the
European Union -- NextGenerationEU, and by PRIN 2022 PNRR project DeLiCE (F53D23009130001).

\bibliographystyle{splncs04}
\bibliography{main}

\iftoggle{arxiv}{
  \clearpage
  \appendix
  \section{DeFi compositions}
\label{sec:defi-compositions}

In this section we discuss some typical compositions of DeFi protocols,
comparing their $\epsilon$-composability and MEV non-interference.

\paragraph{Bet}

We have presented in~\Cref{fig:pricebet} a bet contract that relies on an external price oracle
to determine the winner of a bet.
Babel \emph{et al.} already show in~\cite{Babel23clockwork} that $\contract{Bet}$ is not
$0$-composable with an AMM used as a price oracle whenever the AMM is in an unbalanced state.
The same holds for MEV non-interference $\richnonintrel$, as shown in~\Cref{ex:nonint:pricebet}.
Instead, when using an $\contract{Exchange}$ contract like the one in~\Cref{fig:exchange}
as a price oracle, MEV non-interference $\richnonintrel$ holds, because
the adversary cannot affect the exchange rate.
Technically, this can be proved through condition~\ref{th:richnonint:sufficient-conditions:adv-inert}
of~\Cref{th:richnonint:sufficient-conditions}.
On the other hand, $0$-composability between $\contract{Bet}$ and $\contract{Exchange}$
may hold or not depending on the state of $\contract{Exchange}$.
In particular, when the price given by $\contract{Exchange}$ allows the adversary to 
(honestly!) win the bet, the contracts are not $0$-composable.
Indeed, the value extracted by winning the bet contributes to the MEV of
$\contract{Exchange}[\cdots] \mid \contract{Bet}[\cdots]$,
while this value is not extractable in $\contract{Exchange}[\cdots]$.

\paragraph{Best swap}

The $\contract{BestSwap}$ contract in~\Cref{fig:bestswap}
composes two AMMs, providing a $\txcode{swap}$ method that ensures  
the best (\ie, the minimum) swap rate.
Note that $\contract{BestSwap}$ has contract dependencies
on the two called AMMs,
and token dependencies on their underlying tokens.
Nonetheless, by~\Cref{th:richnonint:sufficient-conditions}\ref{th:richnonint:sufficient-conditions:zero-mev} it enjoys MEV non-interference $\richnonintrel$
\wrt the underlying AMMs, because it has zero MEV
(indeed, the contract balance is always zero, so there is nothing to extract).
For similar reasons, also $0$-composability holds.


\begin{figure}
  \begin{lstlisting}[language=txscript,morekeywords={BestSwap,AMM,swap,getTokens,getRate},classoffset=4,morekeywords={a,b,A,Oracle},keywordstyle=\pmvColor,classoffset=5,morekeywords={ETH,t,t0,t1,tout},keywordstyle=\tokColor,classoffset=6,morekeywords={c0,c1,exch},keywordstyle=\cmvColor,frame=single]
contract (*$\contract[c0,c1]{BestSwap}$*) {
  constructor() {
    require c0.getTokens()==c1.getTokens();
    exch=[c0,c1]; (t0,t1)=exch[0].getTokens();
  }
  getTokens() { // token types
    return (t0,t1);
  }
  getRate(t) { // swap rate
    (r0,r1)=(exch[0].getRate(t),exch[1].getRate(t));
    return min(r0,r1);
  }
  swap(?x:t,ymin) {
    tout = t1 if t==t0 else t0; // output token    
    // choose the AMM with the minimum exchange rate 
    i = 0 if exch[0].getRate(t)<exch[1].getRate(t) else 1;
    exch[i].swap(?x:t,ymin);
    sender!#tout:tout;
  }
}
  \end{lstlisting}
  \negcaptionspace
  \caption{A contract to obtain the best swap between two AMMs.}
  \label{fig:bestswap}
\end{figure}

\paragraph{Swap router}

The $\contract{SwapRouter}$ contract in~\Cref{fig:swaprouter}
composes an AMM with token pair $(\tokT[0],\tokT[1])$
with an AMM with token pair $(\tokT[1],\tokT[2])$,
providing a $\txcode{swap}$ method to swap directly
$\tokT[0]$ with $\tokT[2]$.
\Cref{th:richnonint:sufficient-conditions}\ref{th:richnonint:sufficient-conditions:zero-mev}
ensures that $\contract{SwapRouter}$ enjoys MEV non-interference $\richnonintrel$
\wrt the underlying AMMs, because it has zero MEV
(indeed, the contract balance is always zero, so there is nothing to extract).
For similar reasons, also $0$-composability holds.

\begin{figure}
  \begin{lstlisting}[language=txscript,morekeywords={swap,getTokens,getRate},classoffset=4,morekeywords={a,b,A,Oracle},keywordstyle=\pmvColor,classoffset=5,morekeywords={ETH,t,t0,t1,t2,t1b},keywordstyle=\tokColor,classoffset=6,morekeywords={SwapRouter,c0,c1,exch},keywordstyle=\cmvColor,frame=single]
contract (*$\contract[c0,c1]{SwapRouter}$*) {
  constructor() {
    (t0,t1)=c0.getTokens(); (t1b,t2)=c0.getTokens();
    require t1==t1b;
    exch=[c0,c1];
  }
  getTokens() { return (t0,t2) }
  getRate(t) {
    return exch[0].getRate(t) * exch[1].getRate(t)
  }
  swap(?x:t,ymin) {
    if (t==t0) {              // sell t0, buy t2
      exch[0].swap(?x:t0,0);    // sell t0, buy t1
      exch[1].swap(?#t1:t1,0);  // sell t1, buy t2
      require (#t2>=ymin);
      sender!#t2:t2
    }
    else if (t==t2) {         // sell t2, buy t0
      exch[1].swap(?x:t2,0);    // sell t2, buy t1
      exch[0].swap(?#t1:t1,0);  // sell t1, buy t0
      require (#t0>=ymin);
      sender!#t0:t0
    }
    else abort
  }
}
  \end{lstlisting}
  \negcaptionspace
  \caption{A contract to route swap across two AMMs.}
  \label{fig:swaprouter}
\end{figure}

\paragraph{Lending Pool arbitrage}

The $\contract{LPArbitrage}$ contract
in~\Cref{fig:lp-arbitrage}
combines two AMMs and a Lending Pool to
provide a zero-risk arbitrage.
The contract performs atomically the following transactions:
\begin{enumerate}

\item borrows $x:\tokT[0]$ from the Lending Pool;

\item swaps $x:\tokT[0]$ for a certain amount $y:\tokT[1]$ using the first AMM;

\item swaps $y:\tokT[1]$ for a certain amount $x':\tokT[0]$ using the second AMM;

\item repays the loan $x:\tokT[0]$ to the Lending Pool;

\item transfers the gained units $x' - x$ of $\tokT[0]$ to the caller.
  
\end{enumerate}

\Cref{th:richnonint:sufficient-conditions}\ref{th:richnonint:sufficient-conditions:zero-mev}
ensures that $\contract{LPArbitrage}$ enjoys MEV non-interference $\richnonintrel$
\wrt the underlying AMMs and LP, because it has zero MEV
(indeed, the contract balance is always zero, so there is nothing to extract).
Similarly, Babel \emph{et al.}' \mbox{$0$-composability} holds,
since contracts are sender-agnostic.

\begin{figure}
  \begin{lstlisting}[language=txscript,morekeywords={AMM,swap,arbitrage,borrow,repay,getTokens,getToken,getRate},classoffset=4,morekeywords={a,b,A,Oracle},keywordstyle=\pmvColor,classoffset=5,morekeywords={ETH,t0,t1},keywordstyle=\tokColor,classoffset=6,morekeywords={LPArbitrage,c0,c1,lp},keywordstyle=\cmvColor,frame=single]
contract (*$\contract[c0,c1,lp]{LPArbitrage}$*) {
  constructor() {
    (t0,t1)=lp.getTokens();
    require lp.getToken()==t0 && c1.getTokens()==(t0,t1);
  }
  arbitrage(x) {
    lp.borrow(x);       // borrow x:t0
    c0.swap(?x:t0,0);   // sell t0, buy t1
    c1.swap(?#t1:t1,0); // sell t1, buy t0
    lp.repay(?x:t0);    // repay x:t0
    require #t0>0;      // gain is positive
    sender!#t0:t0       // transfer gain to sender
  }
}
  \end{lstlisting}
  \negcaptionspace
  \caption{A contract to arbitrage with a Lending Pool.}
  \label{fig:lp-arbitrage}
\end{figure}

\begin{figure}
  \begin{lstlisting}[language=txscript,morekeywords={getToken,deposit,borrow,accrue,repay,redeem,liquidate},classoffset=4,morekeywords={a,b,A,Oracle},keywordstyle=\pmvColor,classoffset=5,morekeywords={t,ETH},keywordstyle=\tokColor,classoffset=6,morekeywords={LP},keywordstyle=\cmvColor,frame=single]  
contract LP {
  fun X(n) { // exchange rate 1:t = X() mint t
    if M==0 return 1 else return (n+D*Ir)/M
  }
  fun C(a,n) { // collateralization of a
    return (mint[a]*X(n))/(debt[a]*Ir)
  }
  constructor(c_,r_,m_,t_) { 
    require r>1 && m>1;
    Cmin = c_;  // minimum collateralization
    Rliq = r_;  // liquidation bonus
    Ir   = 1;   // interest rate
    Imul = m_;  // interest rate multiplier
    D = 0;      // total debt
    M = 0;      // total minted tokens
    t = t_;     // handled token
  }
  getToken() { return t; }
  deposit(?x:t) { 
    y = x/X(#t-x); // received mint tokens
    mint[origin]+=y; M+=y;
  } 
  borrow(x) { 
    require #t>x;
    sender!x:t; debt[origin]+=x/Ir; D+=x/Ir; 
    require C(origin,#t)>=Cmin;
  } 
  accrue() {
    require origin == Oracle;
    Ir=Ir*Imul;
  }
  repay(?x:t) {
    require debt[origin]*Ir>=x;
    debt[origin]-=x/i;
  }
  redeem(x) { 
    y = x*(#t);
    require mint[origin]>=x && #t>=y;
    sender!y:t; mint[origin]-=x; M-=x;
    require C(origin,#t)>=Cmin
  }
  liquidate(?x:t,b) { 
    y = (x/X(#t-x))*Rliq;
    require debt[b]*Ir>x && C(b,#t-x)<Cmin && mint[b]>=y;
    mint[origin]+=y; mint[b]-=y; debt[b]-=x/Ir; D-=x/Ir;
    require C(b,#t)<=Cmin;
  }  
}
  \end{lstlisting}
  \negcaptionspace  
  \caption{A Lending Pool contract.}
  \label{fig:lp}
\end{figure}

\paragraph{Flash Loan arbitrage}

The contract $\contract{LParbitrage}$ shown before
requires the user to have a collateral in the $\contract{LP}$
in order to perform the $\txcode{borrow}$ action.
Current Lending Pool protocols like \eg Aave also feature a
$\txcode{flashLoan}$ method which allows users to borrow tokens
without a collateral, provided that the loan is repaid in the
same transaction (see~\Cref{fig:flashloan-arbitrage}).
Note that adding the $\txcode{flashLoan}$ method makes the $\contract{LP}$
contract escape from our blockchain model,
since the dependency relation between contracts is not statically defined:
indeed, the call to $\txcode{exec}$ targets a contract
passed as parameter to $\txcode{flashLoan}$.
However, a simple extension of our contract model
would allow to represent flash loans 
while respecting our well-formedness conditions.
In this extension, a method can specify a boolean condition 
that must hold in the state to be committed upon finalizing the transaction.
When the condition does not hold, the transaction is aborted.

We exploit this extension in~\Cref{fig:flashloan-fixed},
where we show an $\contract{Arbitrage}$ contract relying on a flash loan.
Note that the revised contract enjoys the well-formedness conditions
in~\Cref{sec:blockchain}.
\Cref{th:richnonint:sufficient-conditions}\ref{th:richnonint:sufficient-conditions:zero-mev}
ensures that $\contract{LPArbitrage}$ enjoys MEV non-interference $\richnonintrel$
\wrt the underlying AMMs and LP, because it has zero MEV
(indeed, the contract balance is always zero, so there is nothing to extract).
Instead, Babel \emph{et al.}' \mbox{$0$-composability} is only guaranteed
when the $\code{fee}$ paid to the lending pool for a flash loan is zero.
Indeed, when $\code{fee}>0$, using the flash loan increases the balance
of the lending pool: since \cite{Babel23clockwork} does not
assume adversaries to be rich, 
this extra balance in the lending pool could increase their MEV opportunities.

\begin{figure}[t]
  \begin{lstlisting}[language=txscript,morekeywords={flashLoan,exec,AMM,swap,arbitrage,borrow,repay,getTokens,getToken,getRate},classoffset=4,morekeywords={a,b,A,Oracle},keywordstyle=\pmvColor,classoffset=5,morekeywords={ETH,t0,t1,t},keywordstyle=\tokColor,classoffset=6,morekeywords={LP,FlashLoanArbitrage,c0,c1,lp},keywordstyle=\cmvColor,frame=single]
contract LP {
  flashLoan(rcv,amt,t) {
    int oldBal = #t;
    rcv!amt:t;
    rcv.exec(amt);
    require #t >= oldBal + fee;
  }
}
contract (*$\contract[c0,c1,lp]{FlashLoanArbitrage}$*) {  
  exec(amt) {
    c0.swap(?x:t0,0);   // sell t0, buy t1
    c1.swap(?#t1:t1,0); // sell t1, buy t0
    lp!amt;             // repay flash loan
    sender!#t0:t0       // transfer gain to sender
  }
  arbitrage(x,t0) {
    lp.flashLoan(address(this),x,t0);
  }
}
  \end{lstlisting}
  \negcaptionspace
  \caption{Arbitrage with a Flash Loan (violates well-formedness).}
  \label{fig:flashloan-arbitrage}
\end{figure}

\begin{figure}
  \begin{lstlisting}[language=txscript,morekeywords={flashLoan,exec,AMM,swap,arbitrage,borrow,repay,getTokens,getToken,getRate},classoffset=4,morekeywords={a,b,A,Oracle},keywordstyle=\pmvColor,classoffset=5,morekeywords={ETH,t0,t1,t},keywordstyle=\tokColor,classoffset=6,morekeywords={LP,FlashLoanArbitrage,c0,c1,lp},keywordstyle=\cmvColor,frame=single]
contract LP {
  flashLoan(amt,t) {
    int oldBal = #t;
    sender!amt:t;
    requireFinal #t >= oldBal + fee; // checked just before the transaction is finalized
  }
}
contract (*$\contract[c0,c1,lp]{FlashLoanArbitrage}$*) { 
  arbitrage(x,t0) {
    lp.flashLoan(x);
    c0.swap(?x:t0,0);   // sell t0, buy t1
    c1.swap(?#t1:t1,0); // sell t1, buy t0
    lp!x;               // repay flash loan
    sender!#t0:t0       // transfer gain to sender
    // the requireFinal in LP is checked at this point
  }
}
  \end{lstlisting}
  \negcaptionspace
  \caption{Arbitrage with a Flash Loan (enjoys well-formedness).}
  \label{fig:flashloan-fixed}
\end{figure}

  \section{Supplementary material and proofs}
\label{sec:proofs}

\begin{definition}[Richer state]
  \label{def:richer-state}
  We write $\sysS \leq_{\$} \sysSi$ when the state
  $\sysSi$ can be obtained from $\sysS$ by making the wallets larger,
  \ie when $\sysS = \WmvA \mid \cstC$
  and $\sysSi = (\WmvA + \WmvA[\delta]) \mid \cstC$,
  for some $\WmvA$, $\WmvA[\delta]$, and $\cstC$. 
\end{definition}

\begin{definition}[Wallet-monotonicity]
  \label{def:wallet-monotonicity}
  A blockchain state $\sysS = \WmvA \mid \cstC$ is wallet-monotonic if,
  whenever $\sysS \xrightarrow{\txT} \WmvAi \mid \cstCi$
  for a valid transaction $\txT$,
  then
  $\WmvA + \WmvA[\delta] \mid \cstC \xrightarrow{\txT} \WmvAi + \WmvA[\delta] \mid \cstCi$,
  for all $\WmvA[\delta]$.
\end{definition}

As anticipated in \Cref{sec:blockchain}, we assume that all reachable states are wallet-monotonic.  

\begin{definition}[Wealth of a component]
  \label{def:wealth-sysS}
	Given a state $\sysS = \WmvA \mid \cstC$, $\WmvAi$ contained in $\WmvA$ and $\cstCi$ contained in $\cstC$, 
	we can extend the definition of wealth in the following ways:
	\begin{align*}
		\wealth{}{\sysS} &= \wealth{\AddrU}{\sysS} 
		\\
		\wealth{}{\WmvAi} &= \sum_{\addrA \in \dom{\WmvAi}} \wealth{\addrA}{\sysS} 
		\\
		\wealth{}{\cstCi} &= \sum_{\cmvC \in \cmvOfcst{\cstCi}} \wealth{\cmvC}{\sysS} = \wealth{\cmvOfcst{\cstCi}}{\sysS}
	\end{align*}
\end{definition}
\begin{proofof}{lem:lmev}
	For~\Cref{lem:lmev:mev} we need to prove three equalities.
	\begin{itemize}
		\item[$a.$] $\lmev{\CmvD}{\sysS}{\emptyset} = 0$. Since there are no contracts in $\emptyset$, the identity $\gain{\emptyset}{\sysS}{\TxTS}=0$ holds for any $\TxTS$.
		\item[$b.$] $\lmev{\emptyset}{\sysS}{\CmvC}=0$. 
		We first note that the set $\mall{\emptyset}{\Adv}$ is empty. So, the only $\TxTS\in \mall{\emptyset}{\Adv}^*$ is the empty sequence $\tx{\emptyseq}$. Clearly, it holds $\gain{\CmvC}{\sysS}{\tx{\emptyseq}} = 0$.
		\item[$c.$]
		$\lmev{\CmvU}{\sysS}{\CmvU} \leq \mev{}{\sysS}{}$. We claim that, for any $\TxTS \in \mall{}{\Adv}^*$, we have $\gain{\Adv}{\sysS}{\TxTS} \leq - \gain{\CmvU}{\sysS}{\TxTS}$. Letting $\sysSi$ such that $\sysS \xrightarrow{\TxTS} \sysSi$, we have
		\[
			\wealth{\AddrU}{\sysS} = \wealth{\CmvU}{\sysS} + \wealth{\PmvU \setminus \Adv}{\sysS}+ \wealth{\Adv}{\sysS}
		\]
		and
		\[
			0 = \wealth{\AddrU}{\sysS'}-\wealth{\AddrU}{\sysS} =  \gain{\CmvU}{\sysS}{\TxTS} + \gain{\PmvU \setminus \Adv}{\sysS}{\TxTS}+ \gain{\Adv}{\sysS}{\TxTS}
		\]
		Since a transaction crafted by the adversary can not take tokens from user accounts, the inequality $\gain{\PmvU \setminus \Adv }{\sysS}{\TxTS} \geq 0$ holds, and we have proven the claim $\gain{\Adv}{\sysS}{\TxTS} \leq - \gain{\CmvU}{\sysS}{\TxTS}$.
		Therefore, noting that that $\mall{\CmvU}{\Adv} = \mall{}{\Adv}$, we have
		\begin{align*}
			\lmev{\CmvU}{\sysS}{\CmvU} &= \max \setcomp{ -\gain{\CmvU}{\sysS}{\TxTS} }{ \TxTS \in \mall{\CmvU}{\Adv}^* } 
			\\
			& \geq \max\setcomp{\gain{\Adv}{\sysS}{\TxTS}}{ \TxTS \in \mall{}{\Adv}^* } = \mev{}{\sysS}{}.
		\end{align*}
	
	\end{itemize}
	
	For~\Cref{lem:lmev:L-leq-H}, the proof follows directly from $\mall{\CmvD}{\Adv}\subseteq \mall{\CmvDi}{\Adv}$.

	For~\Cref{lem:lmev:monotonicity}, we will denote as $\cstD$ the contracts that are in $\cstCi$ but not in $\cstC$. In general, we cannot write $\cstCi = \cstC \mid \cstD$, since we have no guarantee that the contracts in $\cstD$ are composed after $\cstC$. So, we write $\cstCi = \cstC \uplus \cstD$. \ricnote{Volendo essere pedanti, non è bello usare $=$, dato che $\cstC \uplus \cstD$ può corrispondere a diversi stati $\cstC$}
	To prove our claim, we take a sequence of transaction $\TxTS\in \mall{\CmvD}{\Adv}^*$ that maximizes the loss $-\gain{\CmvC}{\WmvA \mid \cstC}{\TxTS}$, and we show that the contract loss stays the same when $\TxTS$ is executed in $\WmvA \mid \cstCi$. We can assume without loss of generality that $\TxTS$ is valid in $\WmvA \mid \cstC$: in particular, it never calls methods of contracts outside of $\cstC$.
	Therefore, contracts in $\cstD$ are not affected by $\TxTS$, and we have that  $\WmvA \mid \cstC \xrightarrow{\TxTS} \WmvAi \mid \cstCii$, implies $\WmvA \mid \cstCi \xrightarrow{\TxTS} \WmvAi \mid (\cstCii \uplus \cstD)$.
	Finally, to prove our claim that the loss stays constant we note that 
	\begin{align*}
		\gain{\CmvC}{\WmvA \mid \cstCi}{\TxTS} &= \wealth{\CmvC}{\WmvAi \mid (\cstCii \uplus \cstD) } - \wealth{\CmvC}{\WmvA \mid (\cstCi \uplus \cstD) }= 
		\\
		&= \wealth{\CmvC}{\cstCi \uplus \cstD} - \wealth{\CmvC}{\cstCii \uplus \cstD} 			 
		\\
		&=	\wealth{\CmvC}{\cstCii} + \wealth{\CmvC}{\cstD} -  (\wealth{\CmvC}{\cstCi} + \wealth{\CmvC}{\cstD})=
		\\
		&= \wealth{\CmvC}{\cstCii} - \wealth{\CmvC}{\cstC} =
		\\
		&= \wealth{\CmvC}{\WmvAi \mid \cstCii} - \wealth{\CmvC}{\WmvA \mid \cstC} = \gain{\CmvC}{\WmvA \mid \cstC}{\TxTS}.
	\end{align*}
	
	For~\Cref{lem:lmev:garbage} we first prove the equality
	 $\lmev{\CmvD}{\WmvA \mid \cstC}{\CmvC} = \lmev{\CmvD}{\WmvA \mid \cstC}{\CmvC \cap \cmvOfcst{\cstC}}$.
	Note that for any state $\sysS= \WmvA \mid \cstC$ and sequence of transactions $\TxTS$, we have
	\[
		\gain{\CmvC}{\sysS}{\TxTS} = \sum_{\cmvC \in \CmvC} \gain{\cmvC}{\sysS}{\TxTS} 
	\]
	However, any transaction in $\TxTS$ targeting any contract $\cmvCi\notin \cmvOfcst{\cstC}$ will be invalid, and so we will have $\gain{\cmvCi}{\sysS}{\TxTS}=0$, leading to 
	\[
		\gain{\CmvC}{\sysS}{\TxTS} = \sum_{\cmvC \in \CmvC} \gain{\cmvC}{\sysS}{\TxTS} = \sum_{\cmvC \in \CmvC \cap \cmvOfcst{\cstC}} \gain{\cmvC}{\sysS}{\TxTS} = \gain{\CmvC \cap \cmvOfcst{\cstC}}{\sysS}{\TxTS}.
	\]
	Since this holds for any $\TxTS$, the first equality is proven.	
	To prove the second equality $\lmev{\CmvD}{\WmvA \mid \cstC}{\CmvC} = \lmev{\CmvD\cap \cmvOfcst{\cstC}}{\WmvA \mid \cstC}{\CmvC}$, we simply note that a transaction in $\mall{\CmvD}{\Adv}$ that calls any contracts that do not appear in the state $\cstC$ is invalid, and therefore will have no effect on the loss of $\CmvC$.

	For~\Cref{lem:lmev:leq-wealth}, we prove two inequalities:
	\begin{itemize}
		\item[$a.$] Let $\tx{\emptyseq}$ be the empty sequence of transactions. Clearly, $\tx{\emptyseq}\in \mall{\CmvD}{\Adv}^*$, and $\gain{\CmvC}{\sysS}{\tx{\emptyseq}}=0$, hence 
		\[
			0 =   \gain{\CmvC}{\sysS}{\tx{\emptyseq}} \leq \max \setcomp  { -\gain{\CmvC}{\sysS}{\TxTS}} {\TxTS \in \mall{\CmvD}{\Adv}^*} =  \lmev{\CmvD}{\sysS}{\CmvC}.
		\]

		\item[$b.$]
		The amount of tokens that $\CmvC$ can lose can not be greater than the amount of tokens contained in it, meaning that $\gain{\CmvC}{\sysS}{\TxTS} \geq - \wealth{\CmvC}{\sysS}$, which implies $-\gain{\CmvC}{\sysS}{\TxTS} \leq  \wealth{\CmvC}{\sysS}$.
	\end{itemize}
	\qed
\end{proofof}

\begin{example}
  \label{ex:lmev:not-monotonic-on-observed-contracts}
  $\CmvC \subseteq \CmvCi$ does not imply that
  $\lmev{\CmvD}{\sysS}{\CmvC} \leq \lmev{\CmvD}{\sysS}{\CmvCi}$.
  
  Consider the following contracts,
  with $\price{\tokT[0]} = \price{\tokT[1]} = \price{\tokT[2]} = 1$:
  \begin{lstlisting}[language=txscript,morekeywords={f},classoffset=4,morekeywords={a,A,Oracle},keywordstyle=\pmvColor,classoffset=5,morekeywords={t,T0,T1,T2,ETH},keywordstyle=\tokColor,classoffset=6,morekeywords={C0,C1,C2},keywordstyle=\cmvColor]
    contract C0 { f() { sender!5:T0 } }
    contract C1 { f(?5:T0) { sender!1:T1 } }
    contract C2 { f(?1:T1) { sender!100:T2 } }
  \end{lstlisting}
  
  \noindent
  Let $\Adv = \setenum{\pmvM}$, let
  $\CmvC = \setenum{\contract{C2}}$ and
  $\CmvCi = \setenum{\contract{C1},\contract{C2}}$.
  Let
  \[
  \sysS =
  \walu{\pmvM}{0}{\tokT[2]} \mid
  \walpmv{\contract{C0}}{\waltok{5}{\tokT[0]}} \mid
  \walpmv{\contract{C1}}{\waltok{1}{\tokT[1]}} \mid
  \walpmv{\contract{C2}}{\waltok{100}{\tokT[2]}}
  \]
  For both $\CmvC$ and $\CmvCi$, 
  the transactions sequence that maximizes the local MEV is:
  \[
  \TxTS \; = \;
  \pmvM:\contract{C0}.\txcode{f}() \ \pmvM:\contract{C1}.\txcode{f}() \ \pmvM:\contract{C2}.\txcode{f}()
  \]
  which leads to the state:
  \[
  \sysSi =
  \walu{\pmvM}{100}{\tokT[2]} \mid
  \walpmv{\contract{C0}}{\waltok{0}{\tokT[0]}} \mid
  \walpmv{\contract{C1}}{\waltok{5}{\tokT[0]}} \mid
  \walpmv{\contract{C2}}{\waltok{1}{\tokT[1]}}
  \]
  We have that:
  \begin{align*}
    & \lmev{}{\sysS}{\setenum{\contract{C2}}}
    = -\gain{\setenum{\contract{C2}}}{\sysS}{\TxTS}
    = 100-1
    = 99
    \\
    \not\leq \;\;
    & \lmev{}{\sysS}{\setenum{\contract{C1},\contract{C2}}}
    = -\gain{\setenum{\contract{C1},\contract{C2}}}{\sysS}{\TxTS}
    = 101-6 = 95
    \tag*{\qedex}
  \end{align*}
\end{example}

\begin{proofof}{lem:lmev-wallet}
	For~\Cref{lem:lmev-wallet:zero-user}, we let  $\sysS[1] = \WmvA[\Adv] \mid \WmvA \mid \cstC$, and $\sysS[2] = \WmvA[\Adv] \mid \cstC$. By well formedness, the domains of $\WmvA$ and $\WmvA[\Adv]$ are disjoint, therefore $\WmvA[\Adv] \mid \WmvA = \WmvA[\Adv] + \WmvA$. Moreover, the adversary does not control the wallets in $\WmvA$, so a sequence of transactions $\TxTS \in\mall{\CmvD}{\Adv}^*$ is valid in $\sysS[2]$ if and only if it is valid in $\sysS[1]$.

	By the wallet monotonicity assumption (see \Cref{def:wallet-monotonicity}) we have that if $\TxTS$ is valid (in either state), then
	\[
		\sysS[2]\xrightarrow{\TxTS}  \WmvAi \mid \cstCi  
		\implies
		\sysS[1]\xrightarrow{\TxTS}  \WmvAi + \WmvA \mid \cstCi.
	\] 
	Using this, we get
	\begin{align*}
		\gain{\CmvC}{\sysS[2]}{ \TxTS} &= \wealth{\CmvC}{\WmvAi \mid \cstCi } -  \wealth{\CmvC}{\WmvA[\Adv] \mid \cstC} =\wealth{\CmvC}{\cstCi } -  \wealth{\CmvC}{  \cstC} =
		\\
		 &=  \wealth{\CmvC}{\WmvAi+ \WmvA \mid \cstCi } -  \wealth{\CmvC}{\WmvA[\Adv] +\WmvA \mid \cstC} = \gain{\CmvC}{\sysS[1]}{ \TxTS}.
	\end{align*}
	Since the above equality between $\gain{\CmvC}{\sysS[1]}{ \TxTS}$ and $\gain{\CmvC}{\sysS[2]}{ \TxTS}$ holds for any valid sequence $\TxTS \in\mall{\CmvD}{\Adv}^*$, we have proved the claim $\lmev{\CmvD}{\sysS[1]}{\CmvC} =\lmev{\CmvD}{\sysS[2]}{\CmvC}$.

	For~\Cref{lem:lmev-wallet:monotonicity} we take a sequence of transactions $\TxTS\in \mall{\CmvD}{\Adv}^*$ that maximizes the loss of ${\CmvC}$. Without loss of generality, we can assume that $\TxTS$ is valid in $\sysS$.
	By the wallet monotonicity assumption, $\TxTS$ is still valid in $\sysSi$ and it extracts the same amount of tokens from $\CmvC$, giving us a lower bound for $\lmev{\CmvD}{\sysSi}{\CmvC}$.\qed
\end{proofof}

\begin{proofof}{lem:lmev:stability}
	By \Cref{lem:lmev:leq-wealth} of \Cref{lem:lmev}, the expression $\lmev{\CmvD}{\WmvA \mid \cstC }{\CmvC}$ is bounded from above by $\wealth{\CmvC}{\WmvA \mid \cstC}$.
	By \Cref{def:wealth}, the wealth $\wealth{\CmvC}{\WmvA \mid \cstC}$ only depends on the contract state $\cstC$, so it is equal to $\wealth{\CmvC}{\cstC}$.
	Hence, the expression  $\lmev{\CmvD}{\WmvA[\Adv] \mid \cstC }{\CmvC}$ is an integer value bounded from above, so there exists a wallet $\WmvA[\Adv]$ for which the maximum is reached.  
	By \Cref{lem:lmev-wallet:monotonicity} of \Cref{lem:lmev-wallet}, the same MEV is also achieved with larger wallets $\WmvAi[\Adv]$. \qed
\end{proofof}

\begin{lemma}[Basic properties of $\rlmev{}{}{}$]
  \label{lem:rich-lmev}
  For all $\cstC$, $\CmvC,\CmvD \subseteq \CmvU$:
  \begin{enumerate}

  \item \label{lem:rich-lmev:mev}
    $\rlmev{\CmvD}{\cstC}{\emptyset} = \rlmev{\emptyset}{\cstC}{\CmvC} = 0$

  \item \label{lem:rich-lmev:L-leq-H}
    if $\CmvD \subseteq \CmvDi$, then
    $\rlmev{\CmvD}{\cstC}{\CmvC} \leq \rlmev{\CmvDi}{\cstC}{\CmvC}$
	
  \item \label{lem:rich-lmev:monotonicity}
    $\rlmev{\CmvD}{\cstC}{\CmvC} \leq \rlmev{\CmvD}{\cstCi }{\CmvC}$  where $\cstCi \mid_{\dom{\cstC}} = \cstC$

  \item \label{lem:rich-lmev:garbage}
    $\rlmev{\CmvD}{\cstC}{\CmvC} = \rlmev{\CmvD}{\cstC}{\CmvC \cap \cmvOfcst{\cstC}} = \rlmev{\CmvD\cap \cmvOfcst{\cstC}}{\cstC}{\CmvC}$
    
  \item \label{lem:rich-lmev:leq-wealth}
    $0 \leq \rlmev{\CmvD}{\cstC}{\CmvC} \leq \wealth{\CmvC}{\cstC}$
    
  \end{enumerate}
\end{lemma}
\begin{proof}
	Items \ref{lem:rich-lmev:mev}, \ref{lem:rich-lmev:L-leq-H},  \ref{lem:rich-lmev:monotonicity}, and \ref{lem:rich-lmev:garbage} have an analogous statement in \Cref{lem:lmev}, which holds for any wallet state. Due to the stability lemma (\Cref{lem:lmev:stability}) and the definition of $\rlmev{}{}{}$ (\Cref{def:rich-lmev}), the ``rich-adversary'' versions of the statements must also hold. 
	The same reasoning can be carried out for \Cref{lem:rich-lmev:leq-wealth} and its analogous in \Cref{lem:lmev}, since $\wealth{\CmvC}{\cstC}  = \wealth{\CmvC}{\WmvA \mid \cstC}$ for any $\WmvA$.
	\qed
\end{proof}

\begin{definition}[Sender-agnostic contracts]
	\label{def:sender-agnostic}
	A contract is \emph{sender-agnostic} if the effect of calling each of its methods can be decomposed as follows:
	\begin{itemize}
		\item updating the contract states (either directly, other through inner calls);
		\item transferring tokens from and to users and contracts;
		\item transferring tokens to its $\sender$.
	\end{itemize}
	Further, we require that any call with the same arguments and $\origin$, but distinct $\sender$, has the same effect, except for the third item where tokens are transferred to the new sender.
\end{definition}

\begin{proofof}{th:rich-lmev:stripping}
	First note that $\rlmev{\CmvD}{\strip{\cstC}{\CmvC}}{\CmvC} \leq \rlmev{\CmvD}{\cstC}{\CmvC}$ holds by \Cref{lem:rich-lmev:monotonicity} of \Cref{lem:rich-lmev}, so we just need to show that \begin{equation}
		\label{eq:rich-lmev:stripping:thesis}
		\rlmev{\CmvD}{\cstC}{\CmvC} \leq \rlmev{\CmvD}{\strip{\cstC}{\CmvC}}{\CmvC}.
	\end{equation}
	To do so, we consider a $\WmvA[\Adv]$ that maximizes $\lmev{\CmvD}{\WmvA[\Adv] \mid \cstC }{\CmvC}$ (it must exist by \Cref{lem:lmev:stability}), and show that, for any sequence of transactions $\TxTS\in \mall{\CmvD}{\Adv}^*$ that is valid in $\sysS = \WmvA[\Adv] \mid \cstC$, we can find a $\WmvAi[\Adv]$ and $\TxYS\in\mall{\CmvD}{\Adv}^*$ that is valid in $\sysSi = \WmvAi[\Adv] \mid \strip{\cstC}{\CmvC}$ such that $-\gain{\CmvC}{\sysSi}{\TxYS} \geq - \gain{\CmvC}{\sysS}{\TxTS}$, providing the thesis \eqref{eq:rich-lmev:stripping:thesis}.
		
	We will first construct $\WmvAi[\Adv]$.
	Since the adversary may have different aliases, we rewrite $\WmvA[\Adv]$ as the composition  $\wmvA[1]\mid \wmvA[2]\mid \cdots$. Moreover, the transactions in $\TxTS$ are all valid in $\sysS$, so their $\origin$ must be one of the aliases in the composition. The sequence $\TxTS$ is finite, so we can assume without loss of generality that the aliases which are the origin of some transaction are the first $n$ appearing in the composition.
	Finally, we let
	\[
		\WmvAi[\Adv] = \wmvA[1] +\wmvAi \mid \cdots \mid \wmvA[n] + \wmvAi \mid \wmvA[n+1] \mid \cdots
	\]
	where $\wmvAi$ consists of the sum of all tokens that have been transferred during the execution of $\TxTS$ in state $\sysS$ (both directly and due to internal method calls). 
	Note that the tokens of $\wmvAi$ are added only to the wallets that are origin of some transaction in $\TxTS$, so the finite token axiom is still satisfied.
	
	To prove the thesis, we now construct a sequence of transactions $\TxYS$, valid in $\sysSi = \WmvAi[\Adv] \mid \strip{\cstC}{\CmvC}$ with  $\gain{\CmvC}{\sysSi}{\TxYS} = \gain{\CmvC}{\sysS}{\TxTS}$.
	To do so, we consider $\vec{f}$, the sequence of method calls that are performed upon the execution of $\TxTS$ in state $\sysS$. Note that $\vec{f}$ includes methods called directly from a transaction as well as internal calls that are performed from another called method. 
	We now create a subsequence $\vec{g}$ that contains only the calls that are either
	\begin{enumerate}[$(a)$]
		\item due to a transaction of $\TxTS$ directly calling a contract in $\deps{\CmvC}$, or
		\item due to a internal call in which a method of a contract not in $\deps{\CmvC}$ calls a method of a contract in $\deps{\CmvC}$.
	\end{enumerate}
	\begin{claim}[1]
		A method $m$ whose call appears in $\vec{g}$ due to condition $(b)$ belongs to a contract in $\deps{\CmvC} \cap \deps{\CmvD \setminus \deps{\CmvC}}$.
	\end{claim}
	\begin{proof}[of Claim (1)]
		Clearly $m$ is a method of a contract in $\deps{\CmvC}$.
		Moreover, we know that it has been called internally from a method that is not in $\deps{\CmvC}$, and that this call has  originated from a transaction in $\TxTS$. 
		Such a transaction may only call a contract of $\CmvD$ (since $\TxTS \in \mall{\CmvD}{\Adv}^*$), and we know that it is not calling a method of $\deps{\CmvC}$ (since $\deps{\CmvC}$ is closed downward, and $m$ has been called from a method not in $\deps{\CmvC}$). So, the transaction that originated the call to $m$ must have been targeting a contract in $\CmvD \setminus \deps{\CmvC}$, meaning that $m \in \deps{\CmvD \setminus \deps{\CmvC}}$. \qed
	\end{proof}
	We now let $\TxYS$ be the sequence of transactions that directly perform the calls in $\vec{g}$, in the same order, with the same arguments and $\origin$. If any of the original calls transferred some tokens, then the corresponding transaction of $\TxYS$ will provide the same amount of tokens, taking them from the wallet owned by the alias that originated the method call. By the construction of $\WmvAi[\Adv]$ there are always enough tokens to do so.
	\begin{claim}[2]
		$\TxYS$ belongs to $\mall{\CmvD}{\Adv}^*$.
	\end{claim} 
	\begin{proof}[of Claim (2)]
	We have two cases:
	\begin{itemize}
		\item if a method is in $\vec{g}$ due to $(a)$, then it has been called directly from a transaction in $\TxTS$, which belongs to $\mall{\CmvD}{\Adv}$. That transaction can be copied and put into $\TxYS$.
		\item If a method is in $\vec{g}$ due to $(b)$, then by Claim (1) it belongs to a contract in $\deps{\CmvC} \cap \deps{\CmvD \setminus \deps{\CmvC}}$, which is contained in $\CmvD$ by assumption  \ref{th:rich-lmev:stripping:2}. Moreover, the  adversary is able to craft the arguments of that method by simulating the execution of $\TxTS$. This means that the transaction  $\txY \in \TxYS$ calling the method belongs to $\mall{\CmvD}{\Adv}$.
	\end{itemize}
	\end{proof}

	We  now need to show that $\TxYS$ and $\TxTS$ modify the state of contracts in $\strip{\cstC}{\CmvC}$ in the same way.
	Note that methods that are in $\vec{g}$ due to $(b)$ are sender-agnostic due to assumption \ref{th:rich-lmev:stripping:1} and Claim (1). So, the fact that in the execution of $\TxYS$ they are called directly from a transaction, while in the execution of $\TxTS$ they are called from another contract, does not affect the execution of these methods relatively to $\CmvC$.
	More in detail, it is important to realize that the sequence $\vec{h}$ of method calls performed upon the execution of $\TxYS$ contains $\vec{g}$ but does not coincide with it, since it also includes all the internal calls that are performed by methods in $\vec{g}$. In fact, $\vec{h}$ is the subsequence of $\vec{f}$ that contains every call to methods of contracts in $\deps{\CmvC}$.
	For this reason, both $\vec{f}$ and $\vec{h}$ modify the state of contracts in $\strip{\cstC}{\CmvC}$ in the same way; and this implies that $\TxYS$ is valid in $\sysSi$ and that  $\gain{\CmvC}{\sysSi}{\TxYS} = \gain{\CmvC}{\sysS}{\TxTS}$ (since contracts in $\CmvC$ are not affected by the absence of contracts outside of $\strip{\cstC}{\CmvC}$).
	\qed
\end{proofof}

\begin{remark}\label{remark:rich-lmev:stripping:alternative-proof}
	A significantly simpler proof of \Cref{th:rich-lmev:stripping} can be given assuming $\CmvD \subseteq \deps{\CmvC}$. This is a relevant case, since to apply the definition of MEV non-interference we need to calculate the MEV with  $\CmvD=\CmvC$.
	The first part of the proof stays identical, but we are able to construct the wallet $\WmvAi[\Adv]$ and the sequence of transactions $\TxYS$ much more easily.
	Since $\deps{\CmvC}$ is closed downwards w.r.t. $\sqsubseteq$, every transaction in $\mall{\CmvD}{\Adv}$ cannot invoke contracts that are not in  $\deps{\CmvC}$ (neither directly targeting them, nor through internal contract calls).
	This means that when $\TxTS$ is executed in $\sysS$, it only interacts with contract in $\strip{\cstC}{\CmvC}$, hence $\TxTS$ is also valid in $\WmvA[\Adv] \mid \strip{\cstC}{\CmvC}$. For this reason we can simply chose $\WmvAi[\Adv] = \WmvA[\Adv]$ and $\TxYS=\TxTS$. In this case $\gain{\CmvC}{\sysSi}{\TxYS} = \gain{\CmvC}{\sysS}{\TxTS}$, and we have the thesis \eqref{eq:rich-lmev:stripping:thesis}.
\end{remark}

\begin{example}
  \label{ex:rich-lmev:stripping:1-false}
  We show that if the assumption~\ref{th:rich-lmev:stripping:1}
  does not hold, then we may have $\rlmev{\CmvD}{\cstC}{\CmvC} \gneqq \rlmev{\CmvD}{\strip{\cstC}{\CmvC}}{\CmvC}$. 
  Consider the following contracts:
  \begin{lstlisting}[language=txscript,morekeywords={f,g},classoffset=4,morekeywords={a,A,Oracle},keywordstyle=\pmvColor,classoffset=5,morekeywords={t,T,T0,T1,T2,ETH},keywordstyle=\tokColor,classoffset=6,morekeywords={C0,C1},keywordstyle=\cmvColor]
    contract C0 { f() { require sender==C1; sender!5:T } }
    contract C1 { g() { C0.f(); sender!5:T } }
  \end{lstlisting}
  Note that $\contract{C0}$ violates sender-agnosticism,
  since $\sender$ is used to enable the token transfer only to $\contract{C1}$.
  \bartnote{TODO}
\end{example}

\begin{example}
	\label{ex:rich-lmev:stripping:2-false}
	We show that if the assumption $\deps{\CmvC} \cap \deps{\CmvD \setminus \deps{\CmvC}} \subseteq \CmvD$ does not hold, then we may have $\rlmev{\CmvD}{\cstC}{\CmvC} \gneqq \rlmev{\CmvD}{\strip{\cstC}{\CmvC}}{\CmvC}$. 
	The idea behind this counterexample is that by stripping $\cstC$ of dependencies, we may get rid of contracts that are in $\CmvD$, thus restricting the adversary attacking power.
	Consider the following contracts
	\begin{lstlisting}[language=txscript,morekeywords={f,g},classoffset=4,morekeywords={a,A,Oracle},keywordstyle=\pmvColor,classoffset=5,morekeywords={t,T,T0,T1,T2,ETH},keywordstyle=\tokColor,classoffset=6,morekeywords={C0,C1},keywordstyle=\cmvColor]
    contract C0 { f() { sender!5:T } }
    contract C1 { g() { C0.f(); sender!5:T } }
	\end{lstlisting}
	Here, we let $\Adv = \setenum{\pmvM}$, $\CmvC = \setenum{\contract{C0}}$, $\CmvD= \setenum{\contract{C1}}$ and $\cstC = \walpmv{\contract{C0}}{\waltok{5}{\tokT}} \mid \walpmv{\contract{C1}}{\waltok{0}{\tokT}}$. Stripping $\cstC$ from the dependencies of $\CmvC$ gives $\strip{\cstC}{\CmvC} = \walpmv{\contract{C0}}{\waltok{5}{\tokT}}$.
	Clearly, $\rlmev{\CmvD}{\strip{\cstC}{\CmvC}}{\CmvC} = 0$, since there are no valid transactions in $\mall{\CmvD}{\Adv}$ (the only contract that could be targeted has been stripped).
	On the other hand, 
	$\rlmev{\CmvD}{\cstC}{\CmvC}= 5$, since the adversary can run the transaction $\txT = \pmvM:\contract{C1}.\txcode{g}()$ to extract 5 tokens from $\contract{C0}$.
	Note that 
	\begin{align*}
		\deps{\CmvC} =& \setenum{\contract{C0}}
		\\
		\deps{\CmvD} =& \setenum{\contract{C0}, \contract{C1}} = \deps{\CmvD\setminus \deps{\CmvC}}
		\\
		\deps{\CmvC} \cap \deps{\CmvD\setminus \deps{\CmvC}} =& \setenum{\contract{C0}} \not\subseteq \setenum{\contract{C1}}= \CmvD
	\end{align*}
	so the assumption of \Cref{th:rich-lmev:stripping} does not hold.
\end{example}
 
\begin{example}
  \label{cex:th:rich-lmev:future-not-affect-past:non-rich}
  \label{cex:th:nonint:sufficient-conditions:zero-mev:babel}
  This example shows that
  Babel's composability notion does not enjoy
  condition~\ref{th:nonint:sufficient-conditions:zero-mev} of
  \Cref{th:nonint:sufficient-conditions},
  \ie a contract with zero MEV may not be $0$-composable with the context.
  Furthermore, it is a counterexample to a variant of
  \Cref{th:rich-lmev:future-not-affect-past}
  where $\lmev{}{}{}$ is used instead of $\rlmev{}{}{}$.
  
  Let $\price{\tokT} = \price{\tokT[2]} = 1$,
  let $\contract{Exchange1}$ be a contract 
  (like in \Cref{fig:exchange})
  allowing anyone to swap $1:\tokT$ for $2:\tokT[2]$,
  and let the adversary $\pmvM$ have $1:\tokT[2]$ in $\sysS$:
  \[
  \sysS =
  \walu{\pmvM}{1}{\tokT[2]} \mid
  \walpmv{\contract{Exchange1}}{\waltok{2}{\tokT[2]},\code{tin}=\tokT,\code{tout}=\tokT[2],\code{rate}=2}
  \]
  Let $\cstD = \walpmv{\contract{Exchange2}}{\waltok{1}{\tokT},\code{tin}=\tokT[2],\code{tout}=\tokT,\code{rate}=1}$
  be a contract that swaps tokens $\tokT[2]$
  with an equal amount of tokens $\tokT$.
  Note that using $\contract{Exchange2}$ does not alter its wealth,
  since the number of tokens is preserved, and they have the same price.
  Therefore, $\lmev{}{\sysS \mid \cstD}{\setenum{\contract{Exchange2}}} = 0$,
  and so condition~\ref{th:nonint:sufficient-conditions:zero-mev} of
  \Cref{th:nonint:sufficient-conditions} ensures that $\sysS$ is
  MEV non-interfering with $\cstD$.
  Instead, $\cstD$ is \emph{not} $0$-composable
  according to~\cite{Babel23clockwork}:
  indeed, $\mev{}{\sysS}{} = 0 \neq \mev{}{\sysS \mid \cstD}{} = 1$,
  since $\pmvM$ can first swap her $1:\tokT[2]$ for $1:\tokT$
  using $\contract{Exchange2}$, and then swap $1:\tokT$ for $2:\tokT[2]$
  using $\contract{Exchange1}$, with an overall gain of $1:\tokT[2]$.
  This is coherent with the different interpretations of composability:
  for~\cite{Babel23clockwork}, adding a new contract must preserve the global MEV,
  while our notion requires to preserve its \emph{local} MEV.
  \hfill\qedex
\end{example}

We will now two auxiliary notions,
\ie the token types that can be received by $\cstC$ in $\sysS$,
denoted $\intok{\sysS}{\cstC}$,
and those that can be sent, denoted by $\outtok{\sysS}{\cstC}$.
Token independence relies on these notions.

\begin{definition}[Token independence]
  \label{def:token-independence}  
  Let $\sysS = \WmvA \mid \cstC$.
  We define $\intok{\sysS}{\cstC}$ as the set of token types $\tokT$
  such that there exist a sequence of transactions $\TxTS$,
  a transaction $\tx{Y}$,
  and a state $\sysSi$ with $\sysS \xrightarrow{\TxTS \tx{Y}} \sysSi $
  such that $\tx{Y}$ contains a call to a method of a contract
  in $\cstC$ that receives tokens of type $\tokT$.
  Similarly, $\outtok{\sysS}{\cstC}$ contains $\tokT$
  if there are $\TxTS$, $\tx{Y}$, and $\sysSi$
  with $\sysS \xrightarrow{\TxTS \tx{Y}} \sysSi$
  such that $\tx{Y}$ contains a call to a method of a contract
  in $\cstC$ that sends tokens of type $\tokT$.
  
  We then say that $\cstC$, $\cstD$ are
  \keyterm{token independent} in $\sysS = \WmvA \mid \cstC \mid \cstD$ when
  \(
  \intok{\sysS}{\cstC} \cap \outtok{\sysS}{\cstD} = \emptyset = \intok{\sysS}{\cstD} \cap \outtok{\sysS}{\cstC}
  \).
\end{definition}

\begin{proofof}{th:nonint:sufficient-conditions}
	For condition~\ref{th:nonint:sufficient-conditions:zero-mev} we have
	\begin{align*}
		0 &\leq 	\lmev{\cmvOfcst{\cstD}}{\sysS\mid \cstD}{\cmvOfcst{\cstD}} &&\text{\Cref{lem:lmev:leq-wealth} of  \Cref{lem:lmev} }
		\\
		 &\leq \lmev{}{\sysS\mid \cstD}{\cmvOfcst{\cstD}} &&\text{\Cref{lem:lmev:L-leq-H} of  \Cref{lem:lmev} }
		\\
		& = 0 &&\text{hypothesis}
	\end{align*}
	which implies $\lmev{\cmvOfcst{\cstD}}{\sysS\mid \cstD}{\cmvOfcst{\cstD}} = 0$, and therefore $\nonint{\sysS}{\cstD}$.

	For the remaining items, we first notice that the contract independence of $\cstC$ and $\cstD$ implies that $\cmvOfcst{\cstC}\cap \deps{\cstD} = \emptyset$, and therefore that $\cstD$ is stable \wrt $\Adv$ moves on $\cstC$.
	This means that condition~\ref{th:nonint:sufficient-conditions:adv-inert} implies condition~\ref{th:nonint:sufficient-conditions:no-deps}, so below we will only prove condition~\ref{th:nonint:sufficient-conditions:adv-inert}.

	To prove that $\nonint{\sysS}{\cstD}$ under condition~\ref{th:nonint:sufficient-conditions:adv-inert} we will show that:
        \[\lmev{\cmvOfcst{\cstD}}{\sysS\mid \cstD}{\cmvOfcst{\cstD}} \geq \lmev{}{\sysS\mid \cstD}{\cmvOfcst{\cstD}}
        \]
        (the other direction is guaranteed by \Cref{lem:lmev}).
	Take a sequence of transactions $\TxTS \in \mall{}{\Adv}^*$ that maximizes the loss of $\cmvOfcst{\cstD}$ when executed in state $\sysS\mid \cstD$ (we can assume w.l.o.g. that $\TxTS$ is valid in $\sysS\mid \cstD$). 
	We need to show that there is a sequence $\TxYS \in \mall{\cmvOfcst{\cstD}}{\Adv}^*$ that causes a loss of $\cmvOfcst{\cstD}$ greater or equal to the one caused by $\TxTS$.
	We choose $\TxYS$ to be the subsequence of $\TxTS$ comprising only direct calls to $\cmvOfcst{\cstD}$. Clearly, $\TxYS$ is in $\mall{\cmvOfcst{\cstD}}{\Adv}^*$.

	We now show that $\TxYS$ is valid in $\sysS \mid \cstD$.
	First, we note that any transaction of $\TxTS$ targeting contracts in $\cstC$ (\ie the ones that have been removed to form $\TxYS$) does not change the behaviour of $\cstC$ as observed from $\cstD$.
	This is due to the stability of $\cstD$. 
	For this reason, $\TxTS$ and $\TxYS$ both cause the same state updates in $\cstD$, and in particular they perform the same token transfers to and from it.
	To conclude this proof that $\TxYS$ is valid, we note that the adversary must have enough tokens to fund all the calls in $\TxYS$. Indeed, the adversary has enough funds to execute $\TxTS$, and any token that they gain from the discarded transactions cannot be used by calls in $\TxYS$, due to the token independence assumption.
	Since $\TxYS$ is valid, and updates the state of $\deps{\cstD}$ in the same way as $\TxTS$, the loss caused to $\cstD$ must also be the same, so $\lmev{\cmvOfcst{\cstD}}{\sysS\mid \cstD}{\cmvOfcst{\cstD}} \geq \lmev{}{\sysS \mid \cstD}{\cmvOfcst{\cstD}}$, and we conclude the proof. \qed
\end{proofof}

\medskip

Note that $\richnonint{\cstC}{\cstD}$ does not imply that
$\nonint{\WmvA \mid \cstC}{\cstD}$ \emph{for all} $\WmvA$,
because poor adversaries could not have enough tokens to
manipulate $\cstC$ in order to trigger MEV extraction from $\cstD$.
An example is when $\cstC$ and $\cstD$ are
the $\contract{AMM}$ and $\contract{Bet}$ of~\Cref{ex:nonint:pricebet},
where $\pmvM$ is not rich enough to trigger the price volatility.
Of course we cannot deduce $\richnonint{\cstC}{\cstD}$ 
when $\nonint{\WmvA \mid \cstC}{\cstD}$ holds for \emph{some} $\WmvA$.
Indeed, a poor adversary could not be able to call 
some MEV-triggering method in $\cstC$, while a rich one could:
then, we could observe a discrepancy between the restricted and
unrestricted $\rlmev{}{}{}$ which was not visible for $\lmev{}{}{}$.
The following~\namecref{ex:nonint-not-imply-richnonint} shows this  in detail.

\begin{figure}[t]
  \begin{lstlisting}[language=txscript,morekeywords={get,set,f},classoffset=4,morekeywords={a,caller,A,Oracle},keywordstyle=\pmvColor,classoffset=5,morekeywords={t,T,ETH},keywordstyle=\tokColor,classoffset=6,morekeywords={C,D},keywordstyle=\cmvColor,frame=single]
    contract D {
      f() { require C.get()==1; sender!#ETH:ETH }
    }
    contract C {
      get() { return x }  // x initialized to 0
      set(?1:T) { x=1 }
    }
  \end{lstlisting}
  \negcaptionspace  
  \caption{Contracts for~\Cref{ex:nonint-not-imply-richnonint}.}
  \label{fig:nonint-not-imply-richnonint}
\end{figure}

\begin{example}
  \label{ex:nonint-not-imply-richnonint}
  The implication
  ``if $\nonint{\WmvA \mid \cstC}{\cstD}$ then $\richnonint{\cstC}{\cstD}$''
  does not hold.
  For instance, consider the contracts in~\Cref{fig:nonint-not-imply-richnonint},
  let $\WmvA[n] = \walu{\pmvM}{n}{\tokT}$,
  let $\cstC = \walpmv{\contract{C}}{\waltok{0}{\tokT},\code{x}=0}$
  and let $\cstD = \walpmv{\contract{D}}{\waltok{100}{\ETH}}$.
  We have that $\nonint{\WmvA[0] \mid \cstC}{\cstD}$,
  because $\pmvM$ does not have the $1:\tokT$ that is needed
  to extract the $\ETH$ from $\cstD$.
  To study $\richnonint{}{}$, we have that
  $\rlmev{}{\cstC \mid \cstD}{\setenum{\cmvD}} = 100\cdot\price{\ETH}$,
  while the restricted local MEV
  $\rlmev{\setenum{\cmvD}}{\cstC \mid \cstD}{\setenum{\cmvD}}$
  is zero, since $\pmvM$ cannot call
  the $\txcode{set}$ method of $\contract{C}$.
  Then, $\negrichnonint{\cstC}{\cstD}$.
  \hfill\qedex
\end{example}

\begin{proofof}{lem:richnonint-nonint}
	From \Cref{def:rich-lmev} and \Cref{lem:lmev:stability} there are wallets $\WmvA[1]$ and $\WmvA[2]$ such that
	\begin{align}
	\label{eq:proof:lem:richnonint-nonint-1}
	&\forall \WmvA \geq_{\$} \WmvA[1] \quad \rlmev{\cmvOfcst{\cstD}}{\cstC \mid \cstD}{\cmvOfcst{\cstD}} = \lmev{\cmvOfcst{\cstD}}{\WmvA \mid \cstC \mid \cstD}{\cmvOfcst{\cstD}}
	\\
	\label{eq:proof:lem:richnonint-nonint-2}
	&\forall \WmvA \geq_{\$} \WmvA[2] \quad \rlmev{}{\cstC \mid \cstD}{\cmvOfcst{\cstD}} = \lmev{}{\WmvA \mid \cstC \mid \cstD}{\cmvOfcst{\cstD}}.
	\end{align}
	If we let $\WmvA[0] = \sup \setenum{\WmvA[1] , \WmvA[2]}$, then we have that the following chain of implications:
	\begin{align*}
		\richnonint{\cstC}{\cstD} \iff& \text{ definition of $\richnonint{}{}$}
		\\
		\rlmev{\cmvOfcst{\cstD}}{\cstC \mid \cstD}{\cmvOfcst{\cstD}} = \rlmev{}{\cstC \mid \cstD}{\cmvOfcst{\cstD}} \iff& \text{ from  \eqref{eq:proof:lem:richnonint-nonint-1} and \eqref{eq:proof:lem:richnonint-nonint-2}}
		\\ & \text{ and $\WmvA[0] \geq_{\$} \WmvA[1], \WmvA[2]$ }
		\\
		\lmev{\cmvOfcst{\cstD}}{\WmvA[0] \mid \cstC \mid \cstD}{\cmvOfcst{\cstD}} = \lmev{}{\WmvA[0] \mid \cstC \mid \cstD}{\cmvOfcst{\cstD}} \iff& \text{ definition of $\nonint{}{}$}
		\\
		\nonint{\WmvA[0] \mid \cstC}{\cstD} \iff & \text{ from \eqref{eq:proof:lem:richnonint-nonint-1} and \eqref{eq:proof:lem:richnonint-nonint-2}} 
		\\
		\forall \WmvA \geq_{\$}\WmvA[0] \quad \nonint{\WmvA\mid\cstC}{\cstD}. \qquad \; \; & 
	\end{align*}
	\qed
\end{proofof}
\begin{proofof}{th:richnonint:sufficient-conditions}
	For Item \ref{th:richnonint:sufficient-conditions:zero-mev}, we have 
	\begin{align*}
		0 &\leq 	\lmev{\cmvOfcst{\cstD}}{\cstC\mid \cstD}{\cmvOfcst{\cstD}} &&\text{\Cref{lem:rich-lmev:leq-wealth} of  \Cref{lem:rich-lmev} }
		\\
		 &\leq \lmev{}{\cstC\mid \cstD}{\cmvOfcst{\cstD}} &&\text{\Cref{lem:rich-lmev:L-leq-H} of  \Cref{lem:rich-lmev} }
		\\
		& = 0 &&\text{hypothesis}
	\end{align*}
	which implies $\rlmev{\cmvOfcst{\cstD}}{\cstC}{\cmvOfcst{\cstD}} = 0$, and therefore $\richnonint{\cstC}{\cstD}$.
	Like in \Cref{th:nonint:sufficient-conditions}, we will only prove Item \ref{th:richnonint:sufficient-conditions:adv-inert}, since it implies Item \ref{th:richnonint:sufficient-conditions:no-deps}.
	To prove that $\richnonint{\cstC}{\cstD}$ under the assumptions of Item \ref{th:richnonint:sufficient-conditions:adv-inert}, we will show that 
	\begin{equation}
		\label{eq:proof:th:richnonint:sufficient-conditions}
		\rlmev{\cmvOfcst{\cstD}}{\cstC \mid \cstD }{\cmvOfcst{\cstD}} \geq \rlmev{}{\cstC \mid \cstD }{\cmvOfcst{\cstD}}
	\end{equation}
	(the other inequality is guaranteed by \Cref{lem:rich-lmev}).
	To do so, we consider a $\WmvA[\Adv]$ that maximizes $\lmev{}{\WmvA[\Adv] \mid \cstC \mid \cstD }{\cmvOfcst{\cstD}}$ (it must exist by \Cref{lem:lmev:stability}), and show that, for any sequence of transactions $\TxTS\in \mall{}{\Adv}^*$ that is valid in $\sysS = \WmvA[\Adv] \mid \cstC\mid \cstD$, we can find a $\WmvAi[\Adv]$ and $\TxYS\in\mall{\cmvOfcst{\cstD}}{\Adv}^*$ that is valid in $\sysSi = \WmvAi[\Adv] \mid \cstC\mid \cstD$ such that $-\gain{\cmvOfcst{\cstD}}{\sysSi}{\TxYS} \geq - \gain{\cmvOfcst{\cstD}}{\sysS}{\TxTS}$, proving \eqref{eq:proof:th:richnonint:sufficient-conditions}.
	We choose $\TxYS$ to be the subsequence of $\TxTS$ comprising only direct calls to $\cmvOfcst{\cstD}$. Clearly, $\TxYS$ is in $\mall{\cmvOfcst{\cstD}}{\Adv}^*$. 
	$\WmvAi[\Adv]$ is constructed like in the proof of \Cref{th:rich-lmev:stripping}, by adding to every wallet of $\WmvA[\Adv]$ that is the origin of some transaction in $\TxTS$ an amount of tokens that is equal to the ones that have been transferred during the execution of $\TxTS$ in state $\sysS$.
	We now show that $\TxYS$ is valid in $\sysSi$.
	First, we note that any transaction of $\TxTS$ targeting contracts in $\cstC$ (\ie the ones that have been removed to form $\TxYS$)
	does not change the behaviour of $\cstC$ as observed from $\cstD$.
	This is due to the stability of $\cstD$. 
	For this reason, $\TxTS$ and $\TxYS$ both cause the same state updates in $\cstD$, and in particular they perform the same token transfers to and from it.
	To conclude this proof that $\TxYS$ is valid, we note that the adversary must have enough tokens to fund all the calls in $\TxYS$. Indeed, the adversary has enough funds to execute $\TxTS$, and the construction of $\WmvAi[\Adv]$ accounts for any of the tokens that the adversary might have gained from the transactions that are in $\TxTS$ but not in $\TxYS$.
	Since $\TxYS$ is valid, and updates the state of $\cstD$ in the same way as $\TxTS$, the loss caused to $\cstD$ must be the same, so $\lmev{\cmvOfcst{\cstD}}{\WmvAi[\Adv] \mid \cstC \mid \cstD}{\cmvOfcst{\cstD}} \geq \lmev{}{\WmvA[\Adv] \mid \cstC \mid \cstD}{\cmvOfcst{\cstD}}$, and we conclude the proof. \qed	
\end{proofof}

\hidden{
Nota 1:
	Credo sia possibile indebolire le ipotesi del teorema 4, dimostrandolo in maniera molto simile alla dimostrazione del Lemma 3B. 
	Il nuovo teorema diventerebbe qualcosa del tipo: 
	\begin{itemize}
	\item $\richnonint{\cstC}{\cstD} \implies \richnonint{\strip{\cstC}{\cmvOfcst{\cstD}}}{\cstD} $ incondizionatamente
	\item  $\richnonint{\strip{\cstC}{\cmvOfcst{\cstD}}}{\cstD} \implies \richnonint{\cstC}{\cstD}$ sotto le condizioni di sender-agnostic che ci sono ora, quindi $\deps{\cstD} \cap \deps{\cmvOfcst{\cstC}\setminus\deps{\cstD}}$ devono essere sender-agnostic.
	\end{itemize}
}

\begin{proofof}{th:richnonint:stripping}
	We start by proving the equality
	\begin{equation*}
		\rlmev{}{\cstC \mid \cstD}{\cmvOfcst{\cstD}} = \rlmev{}{\strip{(\cstC \mid \cstD)}{\cmvOfcst{\cstD}}}{\cmvOfcst{\cstD}}
	\end{equation*}
	which, due to \Cref{lem:rich-lmev:garbage} of \Cref{lem:rich-lmev} is equivalent to proving:
	\begin{equation} \label{eq:proof:th:richnonint:stripping:1}
		\rlmev{\cmvOfcst{(\cstC \mid \cstD)}}{\cstC \mid \cstD}{\cmvOfcst{\cstD}} = \rlmev{\cmvOfcst{(\cstC \mid \cstD)}}{\strip{(\cstC \mid \cstD)}{\cmvOfcst{\cstD}}}{\cmvOfcst{\cstD}}.
	\end{equation}
	By letting $\CmvD = \cmvOfcst{(\cstC \mid \cstD)}$, and 
	\[
		\CmvCi = \deps{\cmvOfcst{\cstD}} \cap \deps{\cmvOfcst{(\cstC \mid \cstD)} \setminus \deps{\cmvOfcst{\cstD}}} = \deps{\cmvOfcst{\cstD}} \cap \deps{\cmvOfcst{\cstC } \setminus \deps{\cmvOfcst{\cstD}}}
	\] 
	we can see that $\CmvCi \subseteq \CmvD$ and that contracts in $\CmvCi$ are sender-agnostic (by assumption). This means that both conditions of \Cref{th:rich-lmev:stripping} are satisfied, and we have proven the equality in \eqref{eq:proof:th:richnonint:stripping:1}.
	
	Using \Cref{th:rich-lmev:stripping} again, we will prove the equality
	\begin{equation} \label{eq:proof:th:richnonint:stripping:2}
		\rlmev{\cmvOfcst{\cstD}}{\cstC \mid \cstD}{\cmvOfcst{\cstD}} = \rlmev{\cmvOfcst{\cstD}}{\strip{(\cstC \mid \cstD)}{\cmvOfcst{\cstD}}}{\cmvOfcst{\cstD}}.
	\end{equation}
	This time, we let $\CmvD = \cmvOfcst{\cstD}$ and $\CmvCi = \deps{\cmvOfcst{\cstD}} \cap \deps{\cmvOfcst{\cstD} \setminus \deps{\cstD} }  =\emptyset$. Since $\CmvCi = \emptyset $ the conditions of \Cref{th:rich-lmev:stripping} are trivially satisfied, and \eqref{eq:proof:th:richnonint:stripping:2} holds. 
	Therefore, the following diagram commutes and we have proven the theorem:
  \[
  \begin{tikzcd}[row sep = large, column sep = normal]
  \rlmev{}{\cstC \mid \cstD}{\cmvOfcst{\cstD}}
  \arrow[r, equal, "\text{\eqref{eq:proof:th:richnonint:stripping:1}}"]
  \arrow[d, equal, "\richnonint{\cstC}{\cstD}"]
  & \rlmev{}{\strip{(\cstC \mid \cstD)}{\cmvOfcst{\cstD}}}{\cmvOfcst{\cstD}}
  \arrow[r, equal, "\text{Def. $\strip{}{}$}"]
  & \rlmev{}{\strip{\cstC}{\cmvOfcst{\cstD}} \mid \cstD}{\cmvOfcst{\cstD}}
  \arrow[d, equal, "\richnonint{\strip{\cstC}{\cmvOfcst{\cstD}}}{\cstD}"] \\
  \rlmev{\cmvOfcst{\cstD}}{\cstC \mid \cstD}{\cmvOfcst{\cstD}}
  \arrow[r, equal, "\text{\eqref{eq:proof:th:richnonint:stripping:2}}"]
  & \rlmev{\cmvOfcst{\cstD}}{\strip{(\cstC \mid \cstD)}{\cmvOfcst{\cstD}}}{\cmvOfcst{\cstD}}
  \arrow[r, equal, "\text{Def. $\strip{}{}$}"]
  & \rlmev{\cmvOfcst{\cstD}}{\strip{\cstC}{\cmvOfcst{\cstD}} \mid \cstD}{\cmvOfcst{\cstD}}
  \end{tikzcd}
  \] 
	\qed
\end{proofof}

\begin{proofof}{cor:richnonint:front-running}
	We can apply \Cref{th:richnonint:stripping} to get 
	\begin{align*}
		\richnonint{\cstC}{\cstD} &\iff \richnonint{\strip{\cstC}{\cmvOfcst{\cstD}}}{\cstD}
		\\
		\richnonint{\cstC \mid\cstC[\Adv]}{\cstD} &\iff \richnonint{\strip{(\cstC \mid \cstC[\Adv])}{\cmvOfcst{\cstD}}}{\cstD}
	\end{align*}
	The assumption $\richnonint{\cstC}{\cstD}$ implies that the state $\cstC \mid \cstD$ is well-formed, which in turn implies that all the the dependencies of $\cstD$ are contained in it.
	For this reason we have $\strip{(\cstC \mid \cstC[\Adv])}{\cmvOfcst{\cstD}} = \strip{\cstC}{\cmvOfcst{\cstD}}$, so we have $\richnonint{\cstC \mid\cstC[\Adv]}{\cstD} \iff 	\richnonint{\cstC}{\cstD}$ and the corollary is proven.
\end{proofof}

\begin{lemma}
  \label{lem:richnonint:mid-L}
  If $\richnonint{\cstC}{\cstD}$ then $\richnonint{\cstCi \mid \cstC}{\cstD}$    
\end{lemma}
\begin{proof}
	The hypothesis $\richnonint{\cstC}{\cstD}$ can be restated as follows: for any $\WmvA[1]$ and any sequence $\TxTS\in \mall{}{\Adv}^*$ valid in  $\sysS[1] = \WmvA[1] \mid \cstC \mid \cstD$, we can find $\WmvAi[1]$ and a sequence $\TxTiS\in \mall{\cmvOfcst{\cstD}}{\Adv}^*$ valid in $\sysSi[1] = \WmvAi[1] \mid \cstC \mid \cstD$ such that executing $\TxTiS$ in state $\sysSi[1]$ causes (to contracts in $\cstD$) a loss that is greater or equal to the one caused by executing $\TxTS$ in state $\sysS[1]$.
	Similarly, the thesis $\richnonint{\cstC}{\cstD}$ can be rewritten as follows: for any $\WmvA[2]$ and any sequence $\TxYS\in \mall{}{\Adv}^*$ valid in $\sysS[2]= \WmvA[2] \mid \cstCi \mid \cstC \cstD$, we can find $\WmvAi[2]$ and a sequence $\TxYiS \in \mall{\cmvOfcst{\cstD}}{\Adv}^*$ valid in $\sysSi[2]= \WmvAi[2]\mid  \cstCi \mid \cstC \mid \cstD$, such that executing $\TxYiS$ in state $\sysSi[2]$ causes (to contracts in $\cstD$) a loss that is greater or equal to the one caused by executing $\TxYS$ in state $\sysS[2]$.
	So, we will take $\WmvA[2]$ and $\TxYS$, and construct $\WmvAi[2]$ and $\TxYiS$ to prove our claim.
	First, we construct $\WmvA[1]$ by adding to every wallet of $\WmvA[2]$ that originated a transaction of $\TxYS$ an amount of token equal to the one that are transferred during the execution of $\TxYS$. 
	Then, we construct $\TxTS$ by removing from $\TxYS$ any transaction that targets contracts in $\cstCi$. Notice that since $\cstCi$ is composed before $\cstC$ and $\cstD$, removing these transaction will have no effect on those states. For this reason, the execution of $\TxTS$ in $\sysS[1]$ affects contracts in $\cstD$ in exactly the same way in which the execution of $\TxYS$ in state $\sysS[2]$ does, so we have 
	$\gain{\cmvOfcst{\cstD}}{\sysS[2]}{\TxYS} =\gain{\cmvOfcst{\cstD}}{\sysS[1]}{\TxTS}$.
	Using the hypothesis we can find $\WmvAi[1]$ and $\TxTiS \in \mall{\cmvOfcst{\cstD}}{\Adv}$ with $-\gain{\cmvOfcst{\cstD}}{\TxTiS}{\sysSi[1]} \geq  -\gain{\cmvOfcst{\cstD}}{\TxTS}{\sysS[1]}$.
	To conclude this proof we can simply set $\WmvAi[2]= \WmvAi[1]$ and $\TxYiS = \TxTS$. 
	Indeed, since $\TxTS$ only target contracts of $\cstD$ it will still be valid in a state without $\cstCi$. 
	Moreover, the execution of $\TxYiS$ in $\sysSi[2]$ and that of $\TxTS$ in $\sysSi[1]$ both change the state of contracts in $\cstD$ in the same way, so we have:
	\[
		-\gain{\cmvOfcst{\cstD}}{\sysSi[2]}{\TxYiS} = -\gain{\cmvOfcst{\cstD}}{\sysSi[1]}{\TxTiS} \geq -\gain{\cmvOfcst{\cstD}}{\sysS[1]}{\TxTS} = -\gain{\cmvOfcst{\cstD}}{\sysS[2]}{\TxYS}
	\]
	and the thesis is proven.
	\qed	
	
\end{proof}

\begin{lemma}
  \label{lem:richnonint:erasure}
  Each of the following conditions implies $\richnonint{\cstC}{\cstD}$:
  \begin{enumerate}

  \item \label{lem:richnonint:erasure:R}
    $\richnonint{\cstC \mid \cstCi}{\cstD}$
    and $\deps{\cstD} \cap \cmvOfcst{\cstCi} = \emptyset$,
    
  \item \label{lem:richnonint:erasure:L}
    $\richnonint{\cstCi \mid \cstC}{\cstD}$
    and $\deps{\cstC \mid \cstD} \cap \cmvOfcst{\cstCi} = \emptyset$
  \end{enumerate}  
\end{lemma}
\begin{proof}
	First of all, we note that without the conditions $\deps{\cstD} \cap \cmvOfcst{\cstCi} = \emptyset$ (in case \ref{lem:richnonint:erasure:R}) and $\deps{\cstC \mid \cstD} \cap \cmvOfcst{\cstCi} = \emptyset$ (in case \ref{lem:richnonint:erasure:L}) the state $\cstC \mid \cstD$ is not well formed, and the thesis $\richnonint{\cstC}{\cstD}$ does not make sense.

	Below, we only prove Item \ref{lem:richnonint:erasure:R}, since the proof of Item \ref{lem:richnonint:erasure:L} follows the same steps.
	The assumption $\richnonint{\cstC\mid \cstCi}{\cstD}$ is equivalent to saying that for any $\WmvA[1]$ and any sequence $\TxTS\in \mall{}{\Adv}^*$ valid in  $\sysS[1] = \WmvA[1] \mid \cstC \mid \cstCi \mid \cstD$, there exist $\WmvAi[1]$ and a sequence $\TxTiS\in \mall{\cmvOfcst{\cstD}}{\Adv}^*$ valid in $\sysSi[1] = \WmvAi[1] \mid \cstC \mid \cstCi \mid \cstD$ such that executing $\TxTiS$ in state $\sysSi[1]$ causes (to contracts in $\cstD$) a loss that is greater or equal to the one caused by executing $\TxTS$ in state $\sysS[1]$.
	We want to prove that $\richnonint{\cstC}{\cstD}$, \ie that for any $\WmvA[2]$ and any sequence $\TxYS\in \mall{}{\Adv}^*$ valid in $\sysS[2]= \WmvA[2] \mid \cstC \mid \cstD$, we can find $\WmvAi[2]$ and a sequence $\TxYiS \in \mall{\cmvOfcst{\cstD}}{\Adv}^*$ valid in $\sysSi[2]= \WmvAi[2] \mid \cstC \mid \cstD$, such that executing $\TxYiS$ in state $\sysSi[2]$ causes (to contracts in $\cstD$) a loss that is greater or equal to the one caused by executing $\TxYS$ in state $\sysS[2]$.

	Below, we will assume as given $\WmvA[2]$ and $\TxYS$, and we will construct $\WmvAi[2]$ and $\TxYiS$, proving our claim.
	Since $\TxYS$ is valid in $\sysS[2]$, it does not interact with contracts of $\cstCi$.
	For this reason, if we let $\WmvA[1] = \WmvA[2]$, then $\TxYS$ is valid in $\sysS[2]$, and  $\gain{\cmvOfcst{\cstD}}{\TxYS}{\sysS[1]} = \gain{\cmvOfcst{\cstD}}{\TxYS}{\sysS[2]}$.
	We can now set $\TxTS=\TxYS$. Thanks to our hypothesis, we can find $\WmvAi[1]$ and $\TxTiS \in \mall{\cmvOfcst{\cstD}}{\Adv}$ with $-\gain{\cmvOfcst{\cstD}}{\TxTiS}{\sysSi[1]} \geq  -\gain{\cmvOfcst{\cstD}}{\TxTS}{\sysS[1]}$.
	The sequence $\TxTiS$ is in $\mall{\cmvOfcst{\cstD}}{\Adv}$. Since $\deps{\cstD} \cap \cmvOfcst{\cstCi} = \emptyset$, the execution of $\TxTiS$ does not produce calls to contracts in $\cstCi$, so $\TxTiS$ will be valid in $\sysSi[2]$, where $\WmvAi[2]$ is set to be equal to $\WmvAi[1]$.
	But now, by setting $\TxYiS = \TxTiS$, we have proven our thesis, since
	\[
		-\gain{\cmvOfcst{\cstD}}{\sysSi[2]}{\TxYiS} = -\gain{\cmvOfcst{\cstD}}{\sysSi[1]}{\TxTiS} \geq -\gain{\cmvOfcst{\cstD}}{\sysS[1]}{\TxTS} = -\gain{\cmvOfcst{\cstD}}{\sysS[2]}{\TxYS}
	\]
	\qed	
	\ricnote{Nei commenti si trova una dimostrazione più semplice. Tuttavia fa uso del teorema 1, e quindi deve richiedere che alcuni contratti siano sender-agnostic (come nel teorema 4).
	Una dimostrazione ancora più semplice si può ottenere ricalcando il Lemma B2 e usando direttamente il teorema 4.
	}

\end{proof}
\begin{lemma}
	\label{lem:richnonint:mid-L-zero-mev}
	If $\richnonint{\cstC}{\cstD[1]}$,  $\rlmev{}{\cstC \mid \cstD[1]\mid \cstD[2]}{\cmvOfcst{\cstD[2]}} = 0$, and contracts in $\deps{\cstD[1]}$ are sender-agnostic, then then $\richnonint{\cstC}{\cstD[1]\mid \cstD[2]}$.
\end{lemma}
\begin{proof}
	The hypothesis $\richnonint{\cstC}{\cstD[1]}$ can be restated by saying that for any $\WmvA[1]$ and any sequence $\TxTS\in \mall{}{\Adv}^*$ valid in  $\sysS[1] = \WmvA[1] \mid \cstC \mid \cstD[1]$, there exist  $\WmvAi[1]$ and a sequence $\TxTiS\in \mall{\cmvOfcst{\cstD}}{\Adv}^*$ valid in $\sysSi[1] = \WmvAi[1] \mid \cstC \mid \cstD[1]$ such that executing $\TxTiS$ in state $\sysSi[1]$ causes (to contracts in $\cstD[1]$) a loss that is greater or equal to the one caused by executing $\TxTS$ in state $\sysS[1]$.
	We can also restate the thesis $\richnonint{\cstC}{\cstD[1] \mid \cstD[2]}$ in the following way: for any $\WmvA[2]$ and any sequence $\TxYS\in \mall{}{\Adv}^*$ valid in $\sysS[2]= \WmvA[2] \mid \cstC \mid \cstD[1]\mid \cstD[2]$, there exist $\WmvAi[2]$ and a sequence $\TxYiS \in \mall{\cmvOfcst{(\cstD[1]\mid \cstD[2])}}{\Adv}^*$ valid in $\sysSi[2]= \WmvAi[2] \mid \cstC \mid \cstD[1]\mid \cstD[2]$, such that executing $\TxYiS$ in state $\sysSi[2]$ causes (to contracts in $\cstD[1]\mid \cstD[2]$) a loss that is greater or equal to the one caused by executing $\TxYS$ in state $\sysS[2]$.
	
	So, we will take $\WmvA[2]$ and $\TxYS$, and construct $\WmvAi[2]$ and $\TxYiS$ to prove our claim.
	Since contracts in $\cstD[2]$ have MEV equal to 0, we must have $\gain{\cmvOfcst{\cstD[2]}}{\sysS[2]}{\TxYS} \geq 0$, and so
	\[
		-\gain{\cmvOfcst{(\cstD[1] \mid \cstD[2])}}{\sysS[2]}{\TxYS} =  		-\gain{\cmvOfcst{\cstD[1]}}{\sysS[2]}{\TxYS} -	\gain{\cmvOfcst{\cstD[2]}}{\sysS[2]}{\TxYS} \leq -\gain{\cmvOfcst{\cstD[1]}}{\sysS[2]}{\TxYS}.
	\]
	Now, we construct $\WmvA[1]$ and $\TxTS\in \mall{}{\Adv}^*$. The following construction resembles closely the one in \Cref{th:rich-lmev:stripping}.
	We obtain $\WmvA[1]$ by adding an amount of tokens equal to all the ones that have been transferred during the execution of $\TxYS$ in state $\sysS[2]$ to any wallet of $\WmvA[2]$ that was the origin of some transaction in $\TxYS$.
	To construct $\TxTS$ we look at the sequence of method calls that are performed during the execution of $\TxYS$, and we construct a subsequence by taking only the ones that are: 
	\begin{enumerate*}[$(a)$]
		\item directly performed by a transaction of $\TxYS$ calling a method in $\deps{\cstD}$; or
		\item caused by an internal call in which a method of a contract not in $\deps{\cstD[1]}$ calls a method of a contract in $\deps{\cstD[1]}$.
	\end{enumerate*}
	Then, $\TxTS$ consists of transactions that directly perform the calls in the subsequence, keeping the same arguments and origin.
	Notice that, due to $\cstD[2]$ being composed after $\cstD[1]$, no transaction in $\TxTS$ target a contract of $\cstD[2]$. This means that $\TxTS$ is valid in $\sysS[1] = \WmvA[1]\mid \cstC \mid \cstD[1]$.
	Moreover, since contracts in $\deps{\cstD[1]}$ are sender-agnostic, both $\TxTS$ and $\TxYS$ modify the state of contracts in $\deps{\cstD[1]}$ in the same way, so we have
	\[
		- \gain{\cmvOfcst{\cstD[1]}}{\sysS[1]}{\TxTS} = -\gain{\cmvOfcst{\cstD[1]}}{\sysS[2]}{\TxYS} \geq - \gain{\cmvOfcst{(\cstD[1] \mid \cstD[2])}}{\sysS[2]}{\TxYS}.
	\]
	Now, thanks to the hypothesis $\richnonint{\cstC}{\cstD[1]}$, we know that there are $\WmvAi[1]$ and $\TxTiS \in \mall{\cstD[1]}{\Adv}^*$ valid in $\sysSi[1]$ such that $-\gain{\cmvOfcst{\cstD[1]}}{\sysSi[1]}{\TxTiS}\geq -\gain{\cmvOfcst{\cstD[1]}}{\sysS[1]}{\TxTS}$.

	We can finally conclude the proof by setting $\WmvAi[2] = \WmvAi[1]$ and $\TxYiS = \TxTiS$.
	We note that $\TxTiS$ does not target contracts in $\cstD[2]$ (this is because $\TxTiS$ is valid in $\sysSi[1]$, which does not include contain contracts in $\cstD[2]$) so we must have $\gain{\cmvOfcst{\cstD[2]}}{\sysSi[2]}{\TxYiS} = 0$.
	By combining all the inequalities we have
	\begin{align*}
		-\gain{\cmvOfcst{(\cstD[1] \mid \cstD[2])}}{\sysSi[2]}{\TxYiS}
		= -\gain{\cmvOfcst{\cstD[1]}}{\sysSi[2]}{\TxYiS}
		&=-\gain{\cmvOfcst{\cstD[1]}}{\sysSi[1]}{\TxTiS} 
		\\
		&\geq -\gain{\cmvOfcst{\cstD[1]}}{\sysS[1]}{\TxTS} 
		\geq -\gain{\cmvOfcst{(\cstD[1] \mid \cstD[2])}}{\sysS[2]}{\TxYS}
	\end{align*}
	Which proves our thesis. \qed
\end{proof}

\begin{example}
	\label{ex:richnonint:mid-R}
	\ricnote{volendo si può vedere come caso particolare di B.8}
	$\richnonint{\cstC}{\cstD}$ does not imply $\richnonint{\cstC}{\cstD \mid \cstDi}$.
	This can be simply seen when $\cstD = \emptyset$ (which is trivially non-interfering with any $\cstC$), and  $\cstDi$ is any contract that interferes with $\cstC$.
	For instance, let
	\[
		\cstC = \walpmv{\contract{X}}{\waltok{0}{\tokT}, \code{x}=0} \quad \cstD = \emptyset \quad \cstDi = \walpmv{\contract{C}}{\waltok{1}{\tokT}}
	\]
	where $\contract{X}$ and $\contract{C}$ are defined in \Cref{fig:richnonint:monotonicity}.
	We have that 
	\begin{align*}
		\rlmev{\cmvOfcst{\cstD}}{\cstC \mid \cstD }{\cmvOfcst{\cstD}} &= \rlmev{\emptyset}{\cstC \mid \cstD }{\emptyset} = 0
		\\
		\rlmev{}{\cstC \mid \cstD }{\cmvOfcst{\cstD}} &=  \rlmev{}{\cstC \mid \cstD }{\emptyset} = 0	\end{align*}
	so $\richnonint{\cstC}{\cstD}$. 
	To study the non-interference between $\cstC$ and $\cstD \mid \cstDi$, we note that the \emph{restricted} local MEV of $\contract{C}$ is 0: indeed, if the adversary cannot call methods of $\contract{X}$, then they have no way to change the value of the variable $\code{x}$, and the method $\contract{C}.\txcode{f}()$ will never be enabled. On the other hand, if the adversary has access to $\contract{X}$, she can perform the following sequence of transactions: $\pmvM: \contract{X}.\txcode{set}(1)\ \pmvM: \contract{C}.\txcode{f}()$, extracting $\waltok{1}{\tokT}$ from $\contract{C}$. 
	Therefore, $\negrichnonint{\cstC}{\cstD \mid \cstDi}$.
	Note that since the methods of $\cstDi$ do not depend on $\cstD$, we can change their order while preserving the well formedness of the state, so $\negrichnonint{\cstC}{\cstDi \mid \cstD}$. 
\end{example}
	
	\ricnote{Se $\cstC$ viene composto prima di $\cstD$, può un contratto $\cmvC$ in $\cstC$ fare riferimento (hardcoded, non chiamata ai metodi) ad un contratto $\cmvD$ in $\cstD$? Se sì la seguente congettura è falsa, e richiede di aggiungere l'ipotesi di sender-agnostic in C.}

	Se $\deps{\cstDi} \cap \cmvOfcst{\cstC}$ sono sender-agnostic
	\[\richnonint{\cstC}{\cstD} \text{ and } \rlmev{}{\cstC \mid \cstD \mid \cstDi}{\cstDi} = 0 \text{ implies } \richnonint{\cstC }{\cstD \mid\cstDi}\]
 
\begin{example}
	\label{ex:richnonint:mid-L}
	$\richnonint{\cstC}{\cstD \mid \cstDi}$ does not imply  $\richnonint{\cstC}{\cstD}$.
	For instance, let
	\[
		\cstC = \walpmv{\contract{X}}{\waltok{0}{\tokT}, \code{x}=0} \quad \cstD = \walpmv{\contract{C}}{\waltok{1}{\tokT}} \quad \cstDi = \walpmv{\contract{ForwardX}}{\waltok{0}{\tokT}} 
	\]
	using the contracts in \Cref{fig:richnonint:monotonicity}.
	From the previous example we know that $\negrichnonint{\cstC}{\cstD}$. However, $\richnonint{\cstC}{\cstD \mid \cstDi}$. Indeed, it does not matter if the adversary has direct access to $\contract{X}$, since it can always call $\contract{ForwardX}$ to maximize the loss of $\cstD \mid \cstDi$. 
	More specifically, the sequence of transactions $\pmvM : \contract{ForwardX}.\txcode{set\_x(1)} \ \pmvM: \cmvC.\txcode{f}()$ causes a loss of $\waltok{1}{\tokT}$ in $\cmvC$, so
	\[
		\rlmev{\cmvOfcst{(\cstD \mid \cstDi)}}{\cstC \mid \cstD \mid \cstDi }{\cmvOfcst{(\cstD \mid \cstDi)}} = 1.
	\]
	The unrestricted MEV is also equal to 1 (it cannot be greater, since $\wealth{}{\cstD \mid \cstDi} = 1$), so we have 
	\[
		\rlmev{\cmvOfcst{(\cstD \mid \cstDi)}}{\cstC \mid \cstD \mid \cstDi }{\cmvOfcst{(\cstD \mid \cstDi)}} = 1 = \rlmev{}{\cstC \mid \cstD \mid \cstDi }{\cmvOfcst{(\cstD \mid \cstDi)}}
	\]
	and $\richnonint{\cstC}{ \cstD \mid \cstDi}$.
	Since there are no calls from $\cstDi$ to $\cstD$, we can also change their order while preserving the well formedness of the state, and conclude that
	$\richnonint{\cstC}{ \cstDi \mid \cstD}$.
\end{example}

\begin{figure}[t]
	\begin{lstlisting}[language=txscript,morekeywords={f,get,set,get\_x, set\_x},classoffset=4,morekeywords={a,A,Oracle},keywordstyle=\pmvColor,classoffset=5,morekeywords={t,T,T0,T1,T2,ETH},keywordstyle=\tokColor,classoffset=6,morekeywords={Var,C,X,ForwardX},keywordstyle=\cmvColor,frame=single]
contract X { 
    constructor { x=0 }
    get() { return x }
    set(y) { x = y }
}
contract C{ 
    f(){ require X.get()==1; sender!1:T }
}
contract ForwardX{
    get_x() { return X.get() }
    set_x(y) { X.set(y) }
}
\end{lstlisting}
	\negcaptionspace  \caption{Contracts for examples in \Cref{ex:richnonint:mid-R} and \Cref{ex:richnonint:mid-L}}
	\label{fig:richnonint:monotonicity}
\end{figure}

\begin{figure}[t]
  \begin{lstlisting}[language=txscript,morekeywords={get,set,drop2,drop3},classoffset=4,morekeywords={a,A,Oracle},keywordstyle=\pmvColor,classoffset=5,morekeywords={ETH},keywordstyle=\tokColor,classoffset=6,morekeywords={Var,Drop},keywordstyle=\cmvColor,frame=single]
contract Var {
  constructor { x=0 }
  get()  { return x }   
  set(y) { if x=0 then x=y }
}
contract Drop {
  constructor { b=0 }
  drop2() { require b==0 && Var.get()==1;
            b=1; sender!2:T }
  drop3() { require b==0 && Var.get()==0;
            b=1; c.set(2); sender!3:T }
}
  \end{lstlisting}
  \negcaptionspace  
  \caption{Contracts for the counterexample in~\Cref{ex:richnonint:union}.}
  \label{fig:richnonint:union}
\end{figure}

\begin{example}
  \label{ex:richnonint:union}
  $\richnonint{\cstC}{}{\cstD[1]}$ and
  $\richnonint{\cstC}{}{\cstD[2]}$
  does not imply
  $\richnonint{\cstC}{}{\cstD[1] \mid \cstD[2]}$.
  Let:
  \[
  \cstC =
  \walpmv{\contract{Var}}{\waltok{0}{\tokT},\code{x}=0}
  \quad
  \cstD[1] =
  \walpmv{\contract{Drop1}}{\waltok{3}{\tokT},\code{b}=0}
  \quad
  \cstD[2] =
  \walpmv{\contract{Drop2}}{\waltok{3}{\tokT},\code{b}=0}
  \]
  using the contracts in~\Cref{fig:richnonint:union}.
  The wealthy adversaries' local MEV of both
  $\contract{Drop1}$ and $\contract{Drop2}$ in $\cstC$ is $3$,
  obtained by calling $\txcode{drop3}$ exactly once.
  This holds both for the unrestricted and restricted MEV,
  hence $\richnonint{\cstC}{}{\cstD[1]}$ and $\richnonint{\cstC}{}{\cstD[2]}$.
  To study the non-interference between $\cstC$ and $\cstD[1] \mid \cstD[2]$,
  note that the \emph{unrestricted} local MEV of
  $\setenum{\contract{Drop1},\contract{Drop2}}$ is $4$:
  this is obtained by the sequence
  \(
  \pmvM:\contract{Var}.\txcode{set}(1) \
  \pmvM:\contract{Drop1}.\txcode{drop2}() \
  \pmvM:\contract{Drop2}.\txcode{drop2}()
  \).
  Instead the \emph{restricted} local MEV of
  $\setenum{\contract{Drop1},\contract{Drop2}}$ is $3$,
  since calling $\contract{Var}$ is forbidden.
  Indeed, $\txcode{drop3}$ can only be called on one of the two contracts.
  Therefore, $\negrichnonint{\cstC}{\cstD[1] \mid \cstD[2]}$.
  \hfill\qedex
\end{example}

\paragraph{Weakening the well-formedness assumptions}
In \Cref{sec:blockchain} we have assumed that contract dependencies are statically determined and that blockchain states always contain all the dependencies of their contracts. 
In this paragraph we provide some arguments that show how these assumptions can be lifted. 
First, note that our definitions of local MEV and of $\rlmev{}{}{}$ need not to be modified: in an ill-formed state we can model every call to a non deployed contract as an abort, and calculate the MEV accordingly.
Similarly, the statements in \Cref{lem:lmev,lem:lmev-wallet,lem:lmev:stability} do not need to be modified to account for ill-formed states.

Note that in ill-formed states, the relation $\sqsubseteq$ may fail to be a partial order, or may even be non statically definable. For \Cref{th:rich-lmev:stripping,th:richnonint:stripping}, we can relax the well-formedness assumption so that $\sqsubseteq$ is statically defined on $\cstD$ and its dependencies, and that it partially orders them. Moreover, all these dependencies must be present in the state.
Indeed, the core of the proof of \Cref{th:rich-lmev:stripping} involves removing all calls to methods that are not in $\deps{\cstD}$ from a generic sequence of transactions. This can be done even when the removed contracts are ill-formed. 
\Cref{lem:richnonint:mid-L,lem:richnonint:erasure,lem:richnonint:mid-L-zero-mev} are handled in a similar way.

Lastly, we consider \Cref{th:nonint:sufficient-conditions,th:richnonint:sufficient-conditions}.
Here requesting well-formedness on all contract states $\cstC \mid \cstD$ seems reasonable, since in practice these theorems can be used to guarantee MEV non-interference on a ``small" state, while a later application of \Cref{th:richnonint:stripping} allows us to prove MEV non-interference in a larger state.
In any case, we can still weaken the well-formedness assumptions, by adapting the definition of contract independence.
If $\cstC$ is an ill-formed contract state, $\sqsubseteq$ may fail to be not statically definable, \eg when a contract in $\cstC$ calls another contract whose address is passed as parameter. In this case, we can still define  $\deps{\cstC}$, by over-approximating the set of actually callable contracts from $\cstC$.
Using this definition we can still talk about contract independence, and adapt the proof of the theorems.

}
{}

\end{document}